%% file: arxiv-main.tex
\documentclass{article}
\usepackage{graphicx} 
\usepackage{amsmath}
\usepackage{amssymb}
\usepackage{mathtools}
\usepackage{algorithm}
\usepackage{algpseudocode}
\usepackage{geometry}
\usepackage{color}
\usepackage{amsthm}
\usepackage[backend=bibtex, style=alphabetic]{biblatex}
\usepackage{tikz}
\usetikzlibrary{arrows.meta,positioning,calc}

\usepackage{hyperref}
\usepackage{cleveref}

\usepackage{enumitem}
\usepackage{bbm}
\usepackage{bm}
\usepackage{float}

\bibliography{ref}

\newcommand{\1}{\mathbbm{1}}

\newlist{assumptions}{enumerate}{1}
\setlist[assumptions,1]{
  label=(A\arabic*),
  ref=A\arabic*,
  leftmargin=*,   
  align=left
}

\crefname{assumption}{Assumption}{Assumptions}
\Crefname{assumption}{Assumption}{Assumptions}

\newcommand{\norm}[1]{\left\lVert#1\right\rVert}

\theoremstyle{plain}
\newtheorem{theorem}{Theorem}[section]
\newtheorem{proposition}[theorem]{Proposition}
\newtheorem{lemma}[theorem]{Lemma}

\theoremstyle{definition}
\newtheorem{definition}[theorem]{Definition}

\newtheorem{example}[theorem]{Example}
\newtheorem{claim}[theorem]{Claim}

\theoremstyle{remark}
\newtheorem{remark}[theorem]{Remark}

\newcommand{\eps}{\varepsilon}
\DeclareMathOperator*{\Eop}{\mathbb{E}}

\allowdisplaybreaks

\newcommand{\E}{\mathbb{E}}

\newcommand{\calA}{\mathcal{A}}

\newcommand{\Var}{\mathsf{Var}}

\newcommand{\Law}{\mathrm{Law}}
\newcommand{\bfu}{\boldsymbol{u}}
\newcommand{\bfy}{\boldsymbol{y}}
\newcommand{\bfv}{\boldsymbol{v}}

\newcommand{\bfh}{\boldsymbol{h}}
\newcommand{\bfb}{\boldsymbol{b}}
\newcommand{\bfalpha}{\boldsymbol{\alpha}}

\newcommand{\bftheta}{\boldsymbol{\theta}}
\newcommand{\bfw}{\boldsymbol{w}}

\newcommand{\Cov}{\mathsf{Cov}}

\newcommand{\R}{\mathbb{R}}
\newcommand{\N}{\mathbb{N}}

\newcommand{\brac}[1]{\left\langle #1\right\rangle}

\newcommand{\bfx}{\bm{x}}

\title{Rigorous Asymptotics for First-Order Algorithms
Through the Dynamical Cavity Method}
\author{Yatin Dandi\thanks{Statistical Physics of Computation Laboratory, École polytechnique fédérale de Lausanne (EPFL) CH-1015 Lausanne. Email: yatin.dandi@epfl.ch.}, David Gamarnik\thanks{Sloan School of Management, 
Operations Research Center and Institute of Data,  Systems and Society (IDSS), MIT. Email: gamarnik@mit.edu. }, Francisco Pernice\thanks{CSAIL and LIDS, MIT. Email: fpernice@mit.edu.}, Lenka Zdeborová\thanks{Statistical Physics of Computation Laboratory, École polytechnique fédérale de Lausanne (EPFL) CH-1015 Lausanne. Email: lenka.zdeborova@epfl.ch.}}
\date{}

\begin{document}

\maketitle

\begin{abstract}%
Dynamical Mean Field Theory (DMFT) provides an asymptotic description of the dynamics of macroscopic observables in certain disordered systems. Originally pioneered in the context of spin glasses by \cite{sompolinsky1982relaxational}, it has since been used to derive asymptotic dynamical equations for a wide range of models in physics, high-dimensional statistics and machine learning. One of the main tools used by physicists to obtain these equations is the \emph{dynamical cavity method}, which has remained largely non-rigorous. In contrast, existing mathematical formalizations have relied on alternative approaches, including Gaussian conditioning, large deviations over paths, or Fourier analysis. In this work, we formalize the dynamical cavity method and use it to give a new proof of the DMFT equations for General First Order Methods, a broad class of dynamics encompassing algorithms such as Gradient Descent and Approximate Message Passing.
\end{abstract}

\noindent\textbf{Keywords:} Dynamical Mean Field Theory, Spin Glass Dynamics, Exact Asymptotics.

\section{Introduction}

Dynamical Mean Field Theory (DMFT) is a framework for analyzing the high-dimensional behavior of dynamics in disordered systems, such as those coming from high-dimensional statistics and machine learning. Originating in the physics of spin glasses \cite{sompolinsky1982relaxational, cugliandolo1993analytical}, DMFT provides closed, finite-dimensional equations for the evolution of macroscopic observables ---empirical and population risks, correlation with a planted signal, etc.--- along trajectories that are themselves high-dimensional. In recent years, these ideas have been used in a wide variety of statistical and machine learning applications, where algorithms such as Gradient Descent, Approximate Message Passing (AMP) and Power Iteration are often described, in suitable asymptotic regimes, by DMFT-type equations.

Despite its success, the methodological approaches to DMFT of different communities remain disparate. In the physics literature, a major derivation technique is the dynamical cavity method, where one isolates a coordinate and does a perturbative expansion of the rest of the system with respect to the effect of this coordinate to calculate how the system ``reacts'' to its presence \cite{mezard1987spin}. This approach is powerful and flexible but non-rigorous. By contrast, existing mathematical approaches have largely relied on different formalisms: Gaussian conditioning arguments \cite{bolthausen2014iterative, bayati2011dynamics, celentano2020estimation, celentano2021high, gerbelot2024rigorous}, pathwise large deviations \cite{arous1995large, arous1997symmetric, guionnet1997averaged, ben2006cugliandolo, grunwald1996sanov}, or Fourier methods \cite{bayati2015universality, jones2024fourier}. Notably, \cite{bao2025leave, han2025entrywise} follow an approach that is conceptually similar to the dynamical cavity method, but their proof works by first analyzing the dynamics of AMP and then reducing other DMFT equations to this case, an approach also used in \cite{celentano2020estimation, celentano2021high, dudeja2024spectral}. The dynamical cavity method instead
treats both of these cases in a unified way: the only difference is that in the case of AMP, the perturbative expansion vanishes after the zeroth order term, while in the general case the first-order term also survives.

This paper gives a rigorous version of the dynamical cavity method and uses it to re-derive the DMFT equations for a broad algorithmic class known as General First Order Methods (GFOM) \cite{celentano2020estimation}. GFOM provides a unified template that captures a wide range of iterative procedures, including (variants of) Gradient Descent, AMP, Power Iteration, and related algorithms that alternate between multiplication by a random matrix and its transpose, and coordinate-wise nonlinearities and external randomness.

\subsection{Gradient Descent for Generalized Linear Models Example}

To introduce the model that we study, we begin by presenting the special case of Gradient Descent for a Generalized Linear Model (GLM), which will motivate the general definition.

\paragraph{GLM Setup.} The setting is the standard one for supervised learning: we have covariate-label pairs $(\bfx_i,y_i)$ for $i=1,\dots,N,$ where $\bfx_i \in \R^d$ and $y_i \in \R.$ We assume that the pairs $(\bfx_i,y_i)$ are i.i.d. across $i$ and that $y_i$ depends on $\bfx_i$ only through its inner product with a planted signal vector $\bftheta \in \R^d$. Without loss of generality, we can assume that there exists a measurable function $g :\R^2\to\R$ (on which we will later  impose regularity conditions) and a noise vector $\bfw \sim N(0, I_N)$ such that
\begin{align*}
    \bfy &= g(X\bftheta ; \bfw).
\end{align*}
Above, $\bfy = (y_1,\dots,y_N)$, $X\in \R^{N\times d}$ is the data matrix with rows $\bfx_i$, and it's understood that $g$ is applied coordinate-wise, i.e., $g(X\bftheta ; \bfw) \in \R^N$ with $g(X\bftheta ; \bfw)_i = g(\brac{\bfx_i,\bftheta}; w_i).$

The goal is to find some $\bfv\in \R^d$ such that the map $\bfx\mapsto \brac{\bfx, \bfv}$ is ``close'' to $\bfx\mapsto \brac{\bfx, \bftheta}.$ For this purpose, given some (appropriately regular) loss function $\ell:\R^2\to\R,$ we aim to minimize the empirical risk
\begin{align*}
    R(\bfv) &= \sum_{i=1}^N \ell(\brac{\bfx_i, \bfv},y_i).
\end{align*}
We do this by Gradient Descent (GD), starting from a random $\bfv^0 \in \R^d$ and iterating
\begin{align}\label{eq:gradient-descent-iteration}
    \bfv^t &= \bfv^{t-1} - \eta \nabla R(\bfv^{t-1})
\end{align}
for some learning rate $\eta>0.$ Note that the gradient of the risk can be written as
\begin{align*}
    \nabla R(\bfv) &= \sum_{i=1}^N \partial_1 \ell(\brac{\bfx_i,\bfv}, y_i) \bfx_i \\
    &= X^T \partial_1 \ell(X\bfv, \bfy).
\end{align*}
Above, $\partial_1\ell$ denotes the partial derivative of $\ell$ with respect to the first coordinate, which maps $\R^2\to\R,$ and it's being applied coordinate-wise to each pair $((X\bfv)_i, y_i)$ for $i=1,\dots,N$ to obtain a vector in $\R^N.$ For the rest of the paper, we use this coordinate-wise notation without explicit mention.

\paragraph{What we Aim to Describe.} Under suitable assumptions (see \Cref{sec:main-result}), the DMFT equations describe the evolution of various macroscopic observables in the iteration \eqref{eq:gradient-descent-iteration}. For example, they can be used to compute the limiting normalized empirical risk at step $t$
\begin{align*}
    r^t&= \lim_{N\to\infty}\frac{1}{N} R(\bfv^t),
\end{align*}
or the overlap between any two GD iterates
\begin{align*}
    C_{s,t}&=\lim_{N\to\infty}\frac{1}{N} \brac{\bfv^s,\bfv^t}
\end{align*}
as well as many other quantities of interest along the GD trajectory. To do this, the key objects are joint empirical measures of the vectors $\bfv^t, \bfy$ and
\begin{align*}
    \bfb^t &= X \bfv^{t-1}.
\end{align*}
Indeed, note that $r^t$ is a functional of the empirical measure $\frac{1}{N}\sum_{i=1}^N\delta_{(b^{t+1}_i,y_i)}$ and $\{C_{s,t}\}_{1\leq s,t\leq T}$ is a functional of the empirical measure $\frac{1}{d}\sum_{j=1}^d\delta_{(v_j^1,\dots,v_j^T)}$. The DMFT equations will give exact asymptotic descriptions of these empirical measures, and hence of $r^t,C_{s,t}$, and any other functional. However, in order for the equations to close, it's convenient to work in a larger space of variables, describing a richer empirical measure object, and then obtaining these more specialized quantities as projections. This is what motivates the formulation of the general model that we analyze. We first write the iteration \eqref{eq:gradient-descent-iteration} in this language, and in the next section present the general model abstractly.

\paragraph{GFOM Formulation of Gradient Descent.} We introduce the additional vectors
\begin{align*}
    \bfu^t &= \partial_1 \ell(X\bfv^{t-1}, \bfy) \\
    &= \partial_1 \ell(\bfb^t, \bfy)  \in \R^N
\end{align*}
and
\begin{align*}
    \bfh^t &= X^T \bfu^t =\nabla R(\bfv^{t-1})\in \R^d.
\end{align*}
Now recall that we have $\bfy = g(X\bftheta,\bfw).$ Defining $\bfb^* = X\bftheta$ and absorbing $\partial_1\ell$ and $g$ into one function, there exists a function $F:\R^3 \to\R$ such that
\begin{align*}
    \bfu^t &= F(\bfb^*, \bfb^t; \bfw).
\end{align*}
Moreover, we have
\begin{align*}
    \bfv^t &= \bfv^{t-1} - \eta \bfh^t,
\end{align*}
so there exists a function $G_t:\R^t\to\R$ such that
\begin{align*}
    \bfv^t &= G_t(\bfh^1,\dots,\bfh^t;\bfv^0).
\end{align*}
Overall, we can re-write the iteration \eqref{eq:gradient-descent-iteration} in the following form.
\begin{itemize}
    \item Initialize $\bftheta,\bfv^0 \in \R^d,\bfb^* = X\bftheta \in \R^N.$
    \item For $t=1,\dots,$ let
    \begin{align}\label{eq:GD-iteration}
    &
    \begin{cases}
        \bfb^t = X\bfv^{t-1} &\in \R^N \\
        \bfu^t = F(\bfb^*,\bfb^t;\bfw) &\in \R^N \\
        \bfh^t = X^T \bfu^t &\in \R^d\\
        \bfv^t =G_t(\bfh^1,\dots,\bfh^t;\bfv^0) &\in \R^d.
    \end{cases}
    \end{align}
\end{itemize}
To close the DMFT equations, it will be useful for us to describe the full joint empirical measures of $\bftheta, \{\bfh^t\}_{t=1}^T, \{\bfv^t\}_{t=0}^T$ and $\bfb^*,\{\bfb^t\}_{t=1}^T, \{\bfu^t\}_{t=1}^T,$ respectively.

\subsection{General First Order Methods}
We now present the General First Order Methods (GFOM) model, first introduced in \cite{celentano2020estimation}, which is a generalization of the iteration \eqref{eq:GD-iteration} that also captures algorithms like Approximate Message Passing, Power Iteration, and many others.
\begin{definition}\label{defn:GFOM}
    Let $X \in \R^{N\times d}, \bftheta, \bfv^0 \in \R^d,\bfb^* = X\bftheta \in \R^N.$ For $t=1,\dots,$ suppose $\bfw^t \in \R^N, \tilde{\bfw}^t \in \R^d$, and $F_t, G_t:\R^{2t + 1}\to \R$. The \emph{GFOM orbit} is the iteration
    \begin{align}\label{eq:GFOM-orbit}
        \begin{cases}
            \bfb^t = X \bfv^{t-1} &\in \R^N \\
            \bfu^t = F_t(\bfb^*, \bfb^1,\dots,\bfb^t; \bfw^1,\dots,\bfw^t) & \in \R^N \\
            \bfh^t = X^T \bfu^t &\in \R^d \\
            \bfv^t = G_t(\bfh^1,\dots,\bfh^t;\bfv^0,\tilde{\bfw}^1, \dots, \tilde{\bfw}^t) & \in \R^d.
        \end{cases}
    \end{align}
    In our notation, we will often write $\bfu^t = F_t(\bfb^*, \bfb^1,\dots,\bfb^t)$ and $\bfv^t = G_t(\bfh^1,\dots,\bfh^t),$ making the dependence on $\bfw,\tilde{\bfw}$ and $\bfv^0$ implicit.
\end{definition}

\subsection{DMFT Equations}\label{sec:dmft-equations}
In this section we present the DMFT equations for the GFOM model of \Cref{defn:GFOM}. These equations will hold under some assumptions on the distribution of the data matrix $X$, the noise vectors $\bfw,\tilde{\bfw},$ the initializations $\bftheta,\bfv^0,$ and the functions $F_t,G_t.$ For now, we state the equations, mentioning these assumptions only as they are needed for the statement. We refer the reader to \Cref{sec:main-result} for a formal statement with all assumptions explicitly stated.

We assume throughout that $\alpha := d/N$ is a fixed constant. Given a time horizon $T \geq 0,$ the DMFT equations will describe a pair of probability measures $\nu_T$ and $\mu_T$ in $\R^{2T+2}$ and $\R^{2T+1},$ respectively. Under appropriate assumptions, we will prove that they are limits of the empirical measures
\begin{align}
    \nu_{T,N} &= \frac{1}{d}\sum_{j=1}^d \delta_{(\theta_j, \{h^t_j\}_{t=1}^T, \{v^t_j\}_{t=0}^T )} \label{eq:nu-T-finite-N}\\
    \mu_{T,N} &= \frac{1}{N}\sum_{i=1}^N \delta_{(b^*_i, \{b^t_i\}_{t=1}^T, \{u^t_i\}_{t=1}^T )},\label{eq:mu-T-finite-N}
\end{align}
respectively (see \Cref{sec:main-result}). We will define $\nu_T$ and $\mu_T$ as the laws of explicit \emph{effective processes} $(\theta, \{h^t\}_{t=1}^T,\{v^t\}_{t=0}^T)$ and $(b^*, \{b^t\}_{t=1}^T,\{u^t\}_{t=1}^T)$, respectively, which we will specify via sampling procedures. Moreover, these processes will be defined inductively in $T.$

We will assume that the entries $(\theta_j,v^0_j)$ are sampled i.i.d. across $j=1,\dots,d$ from a fixed, known measure. Hence, $\nu_0$ must be set equal to this measure for it to be a limit of \eqref{eq:nu-T-finite-N}. Then we set
\begin{align}\label{eq:mu-0-base-case}
    b^* \sim N(0,\alpha \E_{(\theta,v^0)\sim \nu_0}[\theta^2])
\end{align}
and $\mu_0 = \Law(b^*).$ This determines the base case $T=0.$ For $T\geq 1$, to define $\mu_T$ we assume $\nu_{T-1}$ has already been defined. We sample a centered Gaussian vector $(b^*, z^1,\dots,z^T)$ with covariance
\begin{align}\label{eq:nu-T-covariance}
    \Cov(b^*, z^1,\dots,z^T) &= \alpha\cdot \Cov_{\nu_{T-1}} (\theta, v^0,\dots,v^{T-1}),
\end{align}
where the right-hand side means the covariance of $(\theta, v^0,\dots,v^{T-1})$ when $(\theta, \{h^t\}_{t=1}^{T-1},\{v^t\}_{t=0}^{T-1})\sim \nu_{T-1}.$ Independently of $(b^*, z^1,\dots,z^T)$, we sample $(w^1,\dots,w^T) \sim N(0,I_T)$, and for $t=1,\dots,T,$ we set
\begin{align}\label{eq:mu-T-recursion}
\begin{cases}
    b^t = z^t + \alpha \sum_{s=1}^{t-1} \E_{\nu_{T-1}}[\partial_{h^s} v^{t-1}]u^s \\
    u^t = F_t(b^*,b^1,\dots,b^t;w^1,\dots,w^t).
\end{cases}
\end{align}
Finally, we set $\mu_T = \Law(b^*, \{b^t\}_{t=1}^T, \{u^t\}_{t=1}^T).$

To define $\nu_T,$ we assume that $\mu_T$ has already been defined. Here we sample $(v^0,\theta) \sim \nu_0$ and, independently, a centered Gaussian vector $(g^1,\dots,g^T)$ with covariance
\begin{align}\label{eq:mu-T-covariance}
    \Cov(g^1,\dots,g^T) &= \Cov_{\mu_T}(u^1,\dots,u^T).
\end{align}
Then, independently of $v^0,\theta,g^1,\dots,g^T,$ we sample $(\tilde{w}^1,\dots,\tilde{w}^T)\sim N(0,I_T),$ and for $t=1,\dots,T,$ we set
\begin{align}\label{eq:nu-T-recursion}
    \begin{cases}
        h^t = g^t + \E_{\mu_T}[\partial_{b^*}u^t] \theta + \sum_{s=1}^t \E_{\mu_T}[\partial_{b^s}u^t] v^{s-1} \\
        v^t = G_t(h^1,\dots,h^t;v^0,\tilde{w}^1,\dots,\tilde{w}^t).
    \end{cases}
\end{align}
Finally, we set $\nu_T = \Law(\theta, \{h^t\}_{t=1}^T, \{v^t\}_{t=0}^T)$.

\begin{remark}\label{rem:meaning-of-field-derivatives-in-dmft}
    In \eqref{eq:mu-T-recursion} and \eqref{eq:nu-T-recursion}, the expressions $\partial_{h^s} v^{t-1},\partial_{b^*}u^t$ and $\partial_{b^s}u^t$ are defined recursively. For instance $\partial_{z^s}u^t$ is recursively defined as follows:
    \begin{align*}
        \partial_{b^s}u^t &= \sum_{k=1}^t \partial_{k+1}F_t(b^*,b^1,\dots,b^t) \partial_{b^s} b^k \\
        \partial_{b^s} b^k &= \1\{s=k\} + \alpha \sum_{r=s}^{k-1} \E_{\nu_{T-1}}[\partial_{h^{r}} v^{k-1}]\partial_{b^s} u^{r}
    \end{align*}
    where $\partial_{k+1}$ means derivative with respect to the $k+1$'th input variable to $F_t.$ The expressions $\partial_{h^s} v^{t-1}$ and $\partial_{b^*}u^t$ are defined analogously.
\end{remark}

\subsection{Main Result}\label{sec:main-result}
We begin by stating the assumptions needed for our main theorem. Below, we assume $\sigma>0$ is a constant independent of $N.$
\begin{assumptions}
  \item\label{assump:1}
  The entries of $X$ are sampled i.i.d.\ from a $\sigma/\sqrt{N}$-subgaussian measure with mean zero and
  variance $1/N$.
  \item\label{assump:2} The entries of $\bfv^0,\bftheta \in \R^d$ are i.i.d. $(\theta_j,v^0_j)\sim\nu_0$ and independent of $X$, where $\nu_0$ is $\sigma$-subgaussian.
  \item \label{assump:3} The functions $F_t,G_t$ are $L$-Lipschitz, where $L$ is a constant independent of $N.$
  \item \label{assump:4} The noise vectors $\bfw^t,\tilde{\bfw}^t$ are independent of each other and of everything else, with entries i.i.d. $N(0,1).$
\end{assumptions}
As it turns out, the last assumption is without loss of generality. Assuming the noise to be Gaussian is done for convenience. We can now state our main theorem.
\begin{theorem}\label{thm:main}
    Suppose assumptions \eqref{assump:1}-\eqref{assump:4} hold in \Cref{defn:GFOM}. Let $T\geq 0$, and fix $L$-Lipschitz functions $\phi:\R^{2T+2}\to \R$ and $\psi:\R^{2T+1}\to\R.$ Let $\nu_{T,N},\mu_{T,N}$ be as in \eqref{eq:nu-T-finite-N}, \eqref{eq:mu-T-finite-N}, respectively and $\nu_T,\mu_T$ be defined by DMFT equations \eqref{eq:mu-0-base-case}-\eqref{eq:nu-T-recursion}. Then as $N,d\to\infty$ with $d/N=\alpha$ fixed, we have
    \begin{align*}
        \int \phi d\nu_{N,T} &\to \int \phi d\nu_{T} \\
        \int \psi d\mu_{N,T} &\to \int \psi d\mu_{T}
    \end{align*}
    in probability.
\end{theorem}
\begin{remark}
Note that we are imposing only Lipschitzness on the non-linearities $F_t,G_t,$ yet in equations \eqref{eq:mu-T-recursion}, \eqref{eq:nu-T-recursion} there are expectations of derivatives of these functions. These derivatives can be understood in the weak (distributional) sense. Equivalently, one can consider a mollifier $\zeta_{t,\eps}:\R^{2t+1}\to \R$ like the centered Gaussian density in $\R^{2t+1}$ with covariance $\eps I_{2t+1},$ and replace the functions $F_t,G_t$ by $F_{t,\eps}= F_t*\zeta_{t,\eps},G_{t,\eps} = G_t* \zeta_{t,\eps},$ respectively. The equations \eqref{eq:mu-0-base-case}-\eqref{eq:nu-T-recursion} can then be solved to obtain measures $\nu_{T,\eps},\mu_{T,\eps}.$ Finally, the measures $\nu_T,\mu_T$ can be defined as (weak) limits of $\nu_{T,\eps},\mu_{T,\eps},$ respectively, as $\eps\to 0$. The fact that this is well-defined and doesn't depend on the choice of mollifier is a standard exercise.
\end{remark}

\subsection{Notation}

We use bold font to denote quantities of diverging dimension as $N\to\infty$, which we call \emph{high-dimensional}, to distinguish them from quantities of constant (in $N$) dimension, which we call \emph{low-dimensional}. For instance, the vectors of \Cref{defn:GFOM} are high-dimensional, while the effective processes defined in \Cref{sec:dmft-equations} are low-dimensional. Following spin glass terminology, we call the matrix $X$ the \emph{disorder} and its entries the \emph{disorder variables.} We refer to the entries of the vectors in the iteration of \Cref{defn:GFOM} as \emph{dynamical variables}, and specifically to the entries of $\bfb$ and $\bfh$ as \emph{fields.} We use $C^{\infty}_{1b}$ to denote the class of functions from $\R^k\to \R$ for some $k$ which are $C^\infty$ with all partial derivatives of non-zero order uniformly bounded. We use the notation $\lesssim$ to denote inequality up to a multiplicative constant that's independent of $N$. Finally, given a random variable $X$ and $p\geq 1,$ we use the notation $\norm{X}_p := \E[|X|^p]^{1/p}.$

\section{Proof Approach: Dynamical Cavity Method}
In this section we describe our proof approach, which is the main technical novelty of the paper. Our proof can be viewed as a formalization of the dynamical cavity method, a powerful non-rigorous tool coming from physics that allows one to quickly derive DMFT equations like \eqref{eq:mu-0-base-case}-\eqref{eq:nu-T-recursion}.

\subsection{Smooth Approximation}
To carry out the cavity derivation, we will need to assume that all functions we deal with are sufficiently regular. The following proposition allows us to assume this without loss of generality. Its proof, which follows \cite[Proposition 4.1]{wei-kwo-universality}, is deferred to \Cref{sec:smooth-approximation}.

\begin{proposition}\label{prop:smooth-approximation}
    To prove \Cref{thm:main}, it suffices to prove it in the special case where the test functions $\phi,\psi$ and the non-linearities $F_t,G_t$ are all $C^\infty_{1b}.$
\end{proposition}
For the case of $C^\infty_{1b}$ functions, we now state a stronger version of \Cref{thm:main} which better lends itself to the dynamical cavity method proof approach. In \Cref{lem:main-strong-implies-main} in the appendix we show that the theorem below implies \Cref{thm:main}.
\begin{theorem}\label{thm:main-strong}
    Suppose assumptions \eqref{assump:1}-\eqref{assump:4} hold in \Cref{defn:GFOM}. Let $T\geq 0$, $\phi:\R^{2T+2}\to \R$ and $\psi:\R^{2T+1}\to\R.$ Suppose all of $\phi,\psi,F_t,G_t$ are $C^\infty_{1b},$ and let $\nu_T,\mu_T$ be defined by DMFT equations \eqref{eq:mu-0-base-case}-\eqref{eq:nu-T-recursion}. Then as $N,d\to\infty$ with $d/N=\alpha$ fixed, the following hold.
    \begin{enumerate}
        \item For all $i\in [N],j\in [d]$, we have
        \begin{align}
            \lim |\E \phi(\theta_j, \{h_j^t\}_{t=1}^T, \{v_j^t\}_{t=0}^T) - \E_{\nu_T} \phi(\theta, \{h^t\}_{t=1}^T, \{v^t\}_{t=0}^T)| &= 0 \\
            \lim |\E \psi(b_i^*, \{b_i^t\}_{t=1}^T, \{u_i^t\}_{t=1}^T) - \E_{\mu_T} \psi(b^*, \{b^t\}_{t=1}^T, \{u^t\}_{t=1}^T)| &= 0.
        \end{align}
        \item For all $i\neq i' \in [N]$ and $j\neq j'\in [d]$, we have
        \begin{align}
            \lim  \left|\Cov\left( \phi(\theta_j, \{h_j^t\}_{t=1}^T, \{v_j^t\}_{t=0}^T),\; \phi(\theta_{j'}, \{h_{j'}^t\}_{t=1}^T, \{v_{j'}^t\}_{t=0}^T)\right)\right| &= 0 \\
            \lim  \left|\Cov\left( \psi(b_i^*, \{b_i^t\}_{t=1}^T, \{u_j^t\}_{t=1}^T),\; \psi(b_{i'}^*, \{b_{i'}^t\}_{t=1}^T, \{u_{i'}^t\}_{t=1}^T)\right)\right| &= 0.
        \end{align}
    \end{enumerate}
\end{theorem}
In words, the joint laws of the tuples $(\theta_j, \{h_j^t\}_{t=1}^T, \{v_j^t\}_{t=0}^T)$ and $(b_i^*, \{b_i^t\}_{t=1}^T, \{u_i^t\}_{t=1}^T)$ converge to $\nu_T,\mu_T$, respectively, in the sense of $C^\infty_{1b}$ test functions, and different coordinates become asymptotically independent.

\subsection{Dynamical Cavity Method Derivation of DMFT Equations}
In this section we carry out the derivation of the DMFT equations by the dynamical cavity method. A non-rigorous version of this derivation is well-known in physics (see e.g. \cite{mezard1987spin}), and is one of the main ways physicists derive these equations in the first place. We present the method in a way that closely follows our proof of \Cref{thm:main-strong}, which will be carried out in detail in \Cref{sec:main-proof}. In this section we focus on the proof ideas and omit several technical details. For brevity, we only derive the form of the field $b^t$ in \eqref{eq:mu-T-recursion}; the derivation for $h^t$ in \eqref{eq:nu-T-recursion} is analogous, and the equations of $u^t,v^t$ in \eqref{eq:mu-T-recursion}, \eqref{eq:nu-T-recursion} follow from those of $b^t,h^t$, respectively, and \Cref{defn:GFOM}.

\paragraph{Warmup: Gaussian Fields Ansatz.} To motivate the derivation, we begin by considering an (incorrect) ansatz that quickly leads to a wrong but instructive version of equations \eqref{eq:mu-0-base-case}-\eqref{eq:nu-T-recursion}. The fixing of this ansatz will naturally lead us to the cavity method.

From \Cref{defn:GFOM}, we have
\begin{align*}
    b^t_i &= \sum_{j=1}^d X_{ij}v_j^{t-1}.
\end{align*}
Consider the ansatz that the vector $\{X_{ij}\}_{j=1}^d$ is approximately independent of $\bftheta,\bfv^0,\bfv^1,\dots,\bfv^{T}.$ If this is true, conditional on $\bfv^{t-1},$ by CLT, $b_i^t$ is approximately Gaussian with variance
\begin{align*}
    \Var(b_i^t|\bfv^{t-1}) &\approx \frac{1}{N}\sum_{j=1}^d (v_j^{t-1})^2.
\end{align*}
If we further assume that the entries of $\bfv^{t-1}$ are approximately independent, the above will concentrate with respect to the randomness of $\bfv^{t-1}$ and we will have
\begin{align*}
    \Var(b_i^t) &\approx \frac{1}{N}\sum_{j=1}^d\E (v_j^{t-1})^2 \\
    &= \alpha \E (v_1^{t-1})^2,
\end{align*}
where we have used exchangeability of the coordinates of $\bfv^{t-1}.$ Now if we inductively (in $t$) assume that $\Law(v_1^{t-1}) \approx \Law_{\nu_{T-1}}(v^{t-1}),$ this yields
\begin{align*}
    \Var(b_i^t) &\approx \alpha \E_{\nu_{T-1}} (v^{t-1})^2.
\end{align*}
The same argument as above, applied to the whole vector $(b^*_i,b^1_i,\dots,b^T_i),$ yields that this vector is approximately Gaussian with covariance
\begin{align*}
    \Cov(b^*_i,b^1_i,\dots,b^T_i) &\approx \alpha \Cov_{\nu_{T-1}}(\theta,v^0,\dots,v^{T-1}).
\end{align*}
Note that this is the same as \eqref{eq:nu-T-covariance}, except in \eqref{eq:nu-T-covariance} the left-hand side has the Gaussian terms $z^t$ only instead of the full fields $b^t$. What's missing in this derivation is the non-Gaussian terms that appear in $b^t$ in \eqref{eq:mu-T-recursion} in addition to $z^t.$ These must arise due to dependence between $\{X_{ij}\}_{j=1}^d$ and $\bftheta,\bfv^0,\dots,\bfv^T.$

\paragraph{Dynamical Cavity Method Proof Sketch.} How do the entries $v_j^{t-1}$ depend on $\{X_{ij}\}_{j=1}^d$?. Recall that $v_j^{t-1}$ depends on $X$ through the fields $h^s_j$ for $s=1,\dots,t-1$, and from the definition of $h_j^s$ we have
\begin{equation}\label{eq:hup}
\begin{split}
    h_j^s &= \sum_{i'=1}^N X_{i'j} u^s_{i'} \\
    &= X_{ij} u^s_{i} + \sum_{i'\neq i} X_{i'j} u^s_{i'}.
\end{split}
\end{equation}
In other words, $h_j^s$ has an \emph{explicit} dependence on $X_{ij}$; it also depends implicitly on $\{X_{ij'}\}_{j'\in [d]\setminus\{j\}}$ through the variables $\bfu^s.$ The idea of our formal version of the dynamical cavity method is to do a Taylor expansion of the dynamical variable $v_j^{t-1}$ with respect to the disorder variable $X_{ij}$ which is driving this explicit dependence. In a nutshell, the zeroth order term in the expansion (which no longer has this explicit dependence) will give rise to the Gaussian terms $z^t$, the first order term will yield the non-Gaussian terms of \eqref{eq:mu-T-recursion}, and the higher order terms will vanish as $N\to\infty.$ This essentially follows the physicist's dynamical cavity method in a rigorous way. The key technical ingredient in making this rigorous will be to prove sharp bounds on certain derivatives of the variables $v_j,h_j,u_i,b_i;$ see \Cref{sec:derivative-bounds}.

Concretely, making the dependence of $v_j^{t-1}$ on the disorder explicit in the notation, defining $X^{ij}$ to be equal to $X$ on all entries except for $X^{ij}_{ij}=0$, and letting $X^{ij}(\eta) = (1-\eta)X^{ij} + \eta X,$ we have
\begin{align}
    b_i^t &= \sum_{j=1}^d X_{ij}v_j^{t-1}(X)\nonumber \\
    &= \sum_{j=1}^d X_{ij}v_j^{t-1}(X^{ij}) +  \sum_{j=1}^d X_{ij}^2\partial_{X_{ij}}v_j^{t-1}(X^{ij}) +\frac{1}{2} \int_0^1 d\eta(1-\eta) \sum_{j=1}^d X_{ij}^3 \partial^2_{X_{ij}}v_j^{t-1}(X^{ij}(\eta)).\label{eq:b-i-t-expansion}
\end{align}
We will show that the third term is $O(1/\sqrt{N})$ and can therefore be ignored. For the second term, we note that the iteration of \Cref{defn:GFOM} up to time $t-1$ depends on $X_{ij}$ only through the variables $b^*_i,b_i^s$ and $h_j^s$, where $X_{ij}$ appears multiplied by $\theta_j, v_j^{s-1}$ and $u_i^s$, respectively, for $s=1,\dots,t-1$. Hence, by chain rule we can write
\begin{align*}
    \partial_{X_{ij}}v_j^{t-1}(X^{ij}) &= \left(\theta_j\partial_{b_i^*}v_j^{t-1}(X^{ij}) + \sum_{s=1}^{t-1} v_j^{s-1}\partial_{b_i^s}v_j^{t-1}(X^{ij})\right) +  \sum_{s=1}^{t-1} u_i^{s}\partial_{h_j^s}v_j^{t-1}(X^{ij}).
\end{align*}
We will show that the parenthesized term is also $O(1/\sqrt{N})$ and can therefore be ignored. This yields
\begin{align}\label{eq:b-i-approx-expansion}
    b_i^t &\approx \sum_{j=1}^d X_{ij}v_j^{t-1}(X^{ij}) + \sum_{s=1}^{t-1} \left(\sum_{j=1}^d X_{ij}^2\partial_{h_j^s}v_j^{t-1}(X^{ij})\right) u_i^s .
\end{align}
To conclude the proof of \eqref{eq:mu-T-recursion}, two main ingredients remain. First, to show that the parenthesized term concentrates, since note that we have
\begin{align*}
    \E \sum_{j=1}^d X_{ij}^2\partial_{h_j^s}v_j^{t-1}(X^{ij})&= \frac{1}{N} \sum_{j=1}^d \E\partial_{h_j^s}v_j^{t-1}(X^{ij}) \\
    &\approx \alpha \E\partial_{h_1^s}v_1^{t-1} \\
    &\approx \alpha \E_{\nu_{T-1}}\partial_{h_1^s}v_1^{t-1},
\end{align*}
where the last two steps can be justified by derivative bounds and an induction argument, respectively. Second, we must show that the first term in \eqref{eq:b-i-approx-expansion} is approximately Gaussian. This will be done by a Lindeberg argument, where we replace the disorder variables $\{X_{ij}\}_{j=1}^d$ by an independent copy $\{\tilde{X}_{ij}\}_{j=1}^d$ and show that this has an $O(1/\sqrt{N})$ effect on the expectations of $C^{\infty}_{1b}$ functions. Once we do this, the ansatz of independence between $\{\tilde{X}_{ij}\}_{j=1}^d$ and $\bfv^{t-1}$ becomes correct, and the derivation of the warmup above is what yields the form of the Gaussian terms $b^*,z^1,\dots,z^T.$

To conclude this proof sketch, we give some intuition for why the term $\sum_{j=1}^d X_{ij}v_j^{t-1}(X^{ij})$ behaves approximately like the Gaussian terms $(b^*,z^1,\dots,z^T)$ in \eqref{eq:mu-T-recursion} while $\sum_{j=1}^d X_{ij}v_j^{t-1}(X)$ does not. As everything else in this section, this will reduce to certain derivative bounds, which are presented in \Cref{sec:derivative-bounds}. To illustrate this reduction, we show how the covariance of the vector $\{\sum_{j=1}^d X_{ij}v_j^{t-1}(X^{ij})\}_{t=1}^T$ is approximately given by $\alpha \Cov_{\nu_{T-1}}(v^0,\dots,v^{T-1})$ as in \eqref{eq:nu-T-covariance}. While this claim is much weaker than the CLT that is needed to prove the DMFT equations, it allows us to quickly illustrate some of the essential ideas without the need for the longer Lindeberg argument, which is carried out in \Cref{sec:main-proof}.
For $1\leq s,t\leq T,$ we have
\begin{align*}
    \E \left(\sum_{j=1}^d X_{ij}v_j^{s-1}(X^{ij})\right)&\left(\sum_{j=1}^d X_{ij}v_j^{t-1}(X^{ij})\right) \\
    &= \sum_{j=1}^d \E X_{ij}^2 v_j^{s-1}(X^{ij})v_j^{t-1}(X^{ij})  + \sum_{j\neq j'}\E X_{ij}X_{ij'} v_j^{s-1}(X^{ij})v_{j'}^{t-1}(X^{ij'}) \\
    &\approx \frac{1}{N} \sum_{j=1}^d \E  v_j^{s-1}v_j^{t-1}  + \sum_{j\neq j'}\E X_{ij}X_{ij'} v_j^{s-1}(X^{ij})v_{j'}^{t-1}(X^{ij'}) \\
    &\approx \alpha  \Cov_{\nu_{T-1}}(v^{s-1},v^{t-1})+ \sum_{j\neq j'}\E X_{ij}X_{ij'} v_j^{s-1}(X^{ij})v_{j'}^{t-1}(X^{ij'}),
\end{align*}
where the last two steps can be justified by derivative bounds and an induction argument, respectively. What remains is to show that the second term above vanishes. For this, again through derivative bounds, we can use the approximate Gaussian integration by parts, which yields
\begin{align*}
     \sum_{j\neq j'}\E X_{ij}X_{ij'} v_j^{s-1}(X^{ij})v_{j'}^{t-1}(X^{ij'}) &\approx \frac{1}{N^2}\sum_{j\neq j'}\E\partial_{X_{ij}}\partial_{X_{ij'}} (v_j^{s-1}(X^{ij})v_{j'}^{t-1}(X^{ij'})).
\end{align*}
Now comes the crucial step where we use the difference between $X^{ij}, X^{ij'}$ and $X.$ Because e.g. $v_j^{s-1}(X^{ij})$ does not depend on $X_{ij}$, it can be treated as a constant when differentiating with respect to $X_{ij}.$ This lets us write
\begin{align}
    \frac{1}{N^2}\sum_{j\neq j'}\E\partial_{X_{ij}}\partial_{X_{ij'}} (v_j^{s-1}(X^{ij})v_{j'}^{t-1}(X^{ij'})) &= \frac{1}{N^2}\sum_{j\neq j'}\E(\partial_{X_{ij'}} v_j^{s-1}(X^{ij}))(\partial_{X_{ij}}v_{j'}^{t-1}(X^{ij'})) \nonumber\\
    \leq \frac{1}{N^2}&\sum_{j\neq j'}\norm{\partial_{X_{ij'}} v_j^{s-1}(X^{ij})}_2\norm{\partial_{X_{ij}}v_{j'}^{t-1}(X^{ij'})}_2.\label{eq:covariance-cross-term-control}
\end{align}
Recall by \eqref{eq:hup} that $v_j$ has explicit dependence only on $X_{ij}$ among $\{X_{ik}\}_{k\in [d]}$. Hence, since $j'\neq j,$ the dependence of $v_j$ on $X_{ij'}$ is implicit, and in this case our derivative bounds will show that we have $\norm{\partial_{X_{ij'}} v_j^{s-1}(X^{ij})}_2\leq O(1/\sqrt{N}),$ and similarly $\norm{\partial_{X_{ij}}v_{j'}^{t-1}(X^{ij'})}_2\lesssim 1/\sqrt{N}.$ Hence
\begin{align*}
    \E \left(\sum_{j=1}^d X_{ij}v_j^{s-1}(X^{ij})\right)\left(\sum_{j=1}^d X_{ij}v_j^{t-1}(X^{ij})\right) &\approx \alpha  \Cov_{\nu_{T-1}}(v^{s-1},v^{t-1}),
\end{align*}
as desired.

\subsection{Key Technical Theorem: Derivative Bounds}\label{sec:derivative-bounds}
In this section we present the derivative bounds that will let us justify the approximations made in the proof sketch of the last section. As is suggested by the proof sketch, we need to bound derivatives with respect to field variables $b^*_i,b_i^s,h_j^s$ as well as disorder variables $X_{ij}.$ We consider the class of partial derivatives with respect to any multi-set of these variables.
\begin{definition}\label{defn:P-partial-derivatives}
    Given a multi-set $P$ of entries from the vectors $\bfb^*,\bfb^s,\bfh^s, s\in \N,$ and matrix $X$, we let $\partial_P$ denote the partial derivative with respect to all variables in the multi-set. We also let $|P|$ denote the number of variables in $P$, counted with repetition.
\end{definition}
\begin{example}\label{ex:P-example}
    Suppose $P = \{X_{ij}, X_{ij'}, b_{i'}^*, b^3_{i''}, X_{ij}, h_j^1, h_j^1\},$ where all of $i,i',i'',j,j'$ are distinct. Then we have
    \[
    \partial_P u_i^t := \partial_{X_{ij}}^2 \partial_{X_{ij'}} \partial_{b^*_{i'}} \partial_{b_{i''}^3} \partial_{h_j^1}^2 u_i^t.
    \]
    Moreover, we have $|P|=7.$
\end{example}

Next, we must distinguish when a multi-set $P$ affects a variable like $v_j^{t-1}$ ``explicitly'' or ``implicitly.'' For this purpose, we first associate a graph to the multi-set $P.$
\begin{definition}\label{defn:graph-of-P}
    Given a multi-set of variables $P$ as in \Cref{defn:P-partial-derivatives}, we associate a subgraph $G(P) = (V(P) = V_1(P)\sqcup V_2(P), E(P))$ of the complete bipartite graph $K_{d,N}$, with parts $V_1(P)\subseteq [d]$ and $V_2(P)\subseteq [N],$ as follows.
    \begin{itemize}
        \item For each coordinate $j$ of $\bfh^s$ which appears in $P$ for some $s\in \N$, we include $j$ in $V_1(P)$ (without repetition).
        \item For each coordinate $i$ of $\bfb^*$ or $\bfb^s$ which appears in $P$ for some $s\in \N$, we include $i$ in $V_2(P)$ (without repetition).
        \item For each pair of coordinates $i\in [N],j\in [d]$ such that $X_{ij}$ appears in $P$, we include $j$ in $V_2(P),$ $i$ in $V_1(P)$, and the edge $\{i,j\}$ in $E(P)$ (all without repetition).
    \end{itemize}
\end{definition}
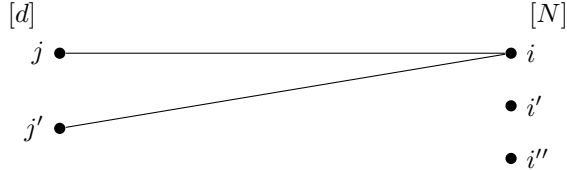
\begin{figure}[h]
\centering
\begin{tikzpicture}[>=stealth, scale=1]

\node at (-0.5,1.5) {$[d]$};
\node[circle, fill=black, inner sep=1.5pt, label=left:$j$]  (j)  at (0,1) {};
\node[circle, fill=black, inner sep=1.5pt, label=left:$j'$] (jp) at (0,0) {};

\node at (6.5,1.5) {$[N]$};
\node[circle, fill=black, inner sep=1.5pt, label=right:$i$]   (i)   at (6,1) {};
\node[circle, fill=black, inner sep=1.5pt, label=right:$i'$]  (ip)  at (6,0.3) {};
\node[circle, fill=black, inner sep=1.5pt, label=right:$i''$] (ipp) at (6,-0.4) {};

\draw (j)  -- (i);
\draw (jp) -- (i);

\end{tikzpicture}
\caption{The graph $G(P)$ associated to the $P$ of \Cref{ex:P-example}.}
\label{fig:P-example}
\end{figure}
We can now state the definition that will capture to our notion of ``explicit'' vs ``implicit'' influence of variables.
\begin{definition}\label{defn:all-reaching}
    Given a derivative multi-set $P$ as in \Cref{defn:P-partial-derivatives} and a vertex $k \in [d]\sqcup [N]$ in the complete bipartite graph $K_{d,N},$ we say $k$ is \emph{all reaching} in $P$ if either $|P|=0$ or $k \in V(P)$ and the graph $G(P)$ is connected. For each $j,t$ say that a variable $v_j^{t-1}$ is all-reaching in $P$ of $j$ is all-reaching in $P.$
\end{definition}
In our proof sketch, a disorder variable ``explicitly'' affecting a dynamical variable corresponds to the dynamical variable being all-reaching in the singleton multi-set corresponding to the disorder variable. Note that no variable is all reaching in the $P$ of \Cref{ex:P-example} because in that case $G(P)$ is disconnected, as can be seen in \Cref{fig:P-example}. On the other hand, if $P = \{X_{ij}, X_{i'j}, b^*_i\}$, then the only all-reaching vertices are $i,i'$ and $j,$ and the all-reaching variables are the ones with those indices.

We are now ready to state our derivative bounds. The bounds will hold under the same assumptions as in \Cref{thm:main-strong}. However, for the applications it will be important that we state them in a more general setting, with assumption 1 slightly relaxed to not require the variance to be exactly $1/N$:
\begin{assumptions}
  \item[(A1')]\makeatletter
    \protected@edef\@currentlabel{A1'}%
  \makeatother
  \label{assump:1-weak}
  The entries of $X$ are sampled i.i.d.\ from a $\sigma/\sqrt{N}$-subgaussian
  measure with mean zero.
\end{assumptions}
The need for this relaxation is that it will make the theorem hold even after some of the entries of $X$ have been set to zero (as we did in the proof sketch), or have been scaled down by a constant in $[0,1]$ as in the third term of \eqref{eq:b-i-t-expansion}.

\begin{theorem}\label{thm:derivative-bounds}
    Suppose assumptions \eqref{assump:1-weak} and \eqref{assump:2}-\eqref{assump:4} hold in \Cref{defn:GFOM}, and $F_t,G_t$ are $C^\infty_{1b}.$  Then, for every $t,p\geq 1$ and $m \geq 0$, there exists a constant $\Gamma_{t,m,p}$ independent of $N$ such that the following hold for every multiset of derivatives $P$ such that $|P|=m$ and for every $i \in [N]$ and $j \in [d]$.
    \begin{enumerate}
        \item For $s=1,\dots,t,$ we have
        \begin{align*}
            \max\{\norm{\partial_P b^*_i}_p,\; \norm{\partial_P b^s_i}_p, \; \norm{\partial_P u^s_i}_p,\; \norm{\partial_P h^s_j}_p,\; \norm{\partial_P v^s_j}_p \} &\leq \Gamma_{t,m,p}.
        \end{align*}
        \item If $i$ is not all reaching in $P,$ then
        \begin{align*}
            \max\{\norm{\partial_P b^*_i}_p,\; \norm{\partial_P b^s_i}_p, \; \norm{\partial_P u^s_i}_p\} &\leq \Gamma_{t,m,p}/\sqrt{N}.
        \end{align*}
        Similarly, if $j$ is not all reaching in $P$, then
        \begin{align*}
            \max\{\norm{\partial_P h^s_j}_p,\; \norm{\partial_P v^s_j}_p \} &\leq \Gamma_{t,m,p}/\sqrt{N}.
        \end{align*}
    \end{enumerate}
\end{theorem}
In words, derivatives always have moments at most $O(1)$, and if the variable we are differentiating is not all-reaching in the differentiating multi-set, then all moments become $O(1/\sqrt{N}).$ Returning to \eqref{eq:covariance-cross-term-control}, note that the derivatives that appear therein have disconnected graphs consisting of three vertices and one edge. The theorem above justifies the fact that their 2-norm is $O(1/\sqrt{N}).$
\begin{remark}
    We believe \Cref{thm:derivative-bounds} is sharp, at least in the sense that for all $m$, there exists $t$ large enough such that the dependence on $N$ of the bounds above corresponding to $t$ and $m$ cannot be improved.
\end{remark}
Our proof of \Cref{thm:derivative-bounds} will be through an inductive argument in $t,$ which bounds the constant $\Gamma_{t,m,p}$ by a function of the constants $\Gamma_{t-1,m',p'},$ for $m',p'$ which are potentially larger than $m,p$. For this to work, cancellations need to be exploited when bounding derivatives of the field variables $\partial_P b^*_i\partial_P b^s_i$ and $\partial_P h_j^s$. This is done by a Taylor expansion which isolates the dependence on each disorder variable. A similar argument was used in \cite{wei-kwo-universality} in the context of proving universality in the disorder for AMP. In \cite{wei-kwo-universality}, a version of part (1) of \Cref{thm:derivative-bounds} was proven in a slightly different model by this strategy. For our purposes, it's crucial to also obtain part (2), as is clear from the proof sketch of the last section. Obtaining this strengthening requires quite a delicate analysis, carried out in \Cref{sec:proof-of-derivative-bounds}.

\section*{Acknowledgments}
The authors thank Giulio Biroli, Gerard Ben Arous, Florent Krzakala, Andrea Montanari for insightful discussions.

This work was started while the author Y. Dandi was visiting the Massachusetts Institute of Technology (MIT), with support from the EPFL Doc.Mobility grant, whose financial assistance is gratefully acknowledged. Early discussions on this work took place in the Simons Laufer Mathematical Sciences Institute (SLMath) in Berkeley, California, during the 2025 program Probability and Statistics of Discrete Structures. The authors are grateful to SLMath for its hospitality and support.

\printbibliography

\input{app-submission-draft}

\end{document}

%% file: app-submission-draft.tex
\appendix
\section{Remaining Proofs}
\subsection{Smooth Approximation: Proof of \Cref{prop:smooth-approximation}}\label{sec:smooth-approximation}
\begin{proof}[of \Cref{prop:smooth-approximation}]
Let $\zeta_t:\R^{2t+1}\to\R$ be a mollifier supported in the unit ball in $\R^{2t+1}$ such that $\int \zeta_t(x)dx = 1,$ and let $\zeta_{t,\eps}(x) = \eps^{-(2t+1)}\zeta_t(x/\eps)$, which is now supported in the $\eps$-ball. Let $F_{t,\eps} = F_t * \zeta_{t,\eps},G_{t,\eps} = G_t * \zeta_{t,\eps},$ and let $\bfb^{t,\eps},\bfu^{t,\eps},\bfh^{t,\eps},\bfv^{t,\eps}$ be the result of running the iteration of \Cref{defn:GFOM} with the same values of $X,\bftheta,\bfv^0,\bfw,\Tilde{\bfw}$ as for $\bfb^t,\bfu^t,\bfh^t$ and $\bfv^t,$ but with $F_{t,\eps},G_{t,\eps}$ in place of $F_t,G_t.$

First, we claim that we have 
\begin{align*}
    \sup_{x}|F_{t,\eps}(x)-F_t(x)| &\leq \eps L \\
    \sup_{x}|G_{t,\eps}(x)-G_t(x)| &\leq \eps L,
\end{align*}
where $L$ is a bound on is the Lipschitz constants of $F_t,G_t$. Indeed, letting $B_\eps$ be the $\eps$-ball, we have 
\begin{align}
    |F_{t,\eps}(x)-F_t(x)| &\leq \int_{B_\eps} \zeta_{t,\eps}(y)|F_t(x-y) - F_t(x)|dy \nonumber\\
    &\leq L\eps. \label{eq:F-eps-Lipschitz-bound}
\end{align}
The proof for $G$ is identical.

Next we claim that $F_{t,\eps},G_{t,\eps}$ have all partial derivatives of non-zero order uniformly bounded. We prove this only for $F_{t,\eps},$ since the proof for $G_{t,\eps}$ is the same. Consider a derivative tuple $\bfalpha \in \N^{2t+1}$, where we use the notation 
\[
\partial^{\bfalpha} F_{t,\eps} := \partial_1^{\alpha_1}\dots\partial_{2t+1}^{\alpha_{2t+1}} F_{t,\eps}.
\]
We assume that $|\bfalpha|:= \sum_{r=1}^{2t+1}\alpha_r \geq 1$, let $r'\in [2t+1]$ be such that $\alpha_{r'}\geq 1$, and let $\bfalpha' = (\alpha_1,\dots,\alpha_{r'}-1,\dots,\alpha_{2t+1}).$ Since $F_t$ is almost-everywhere differentiable, we can write
\begin{align*}
    \partial^{\bfalpha} F_{t,\eps}(x) &= \int \partial^{\bfalpha'}\zeta_{t,\eps}(x-y)\partial_{r'}F_t(y)dy \\
    &= \frac{1}{\eps^{2t+|\bfalpha|}}\int \partial^{\bfalpha'}\zeta_{t}\left(\frac{x-y}{\eps}\right)\partial_{r'}F_t(y)dy \\
    &= \eps^{1-|\bfalpha|}\int \partial^{\bfalpha'}\zeta_{t}(z)\partial_{r'}F_t(x-\eps z)dz.
\end{align*}
Since $\partial^{\bfalpha'}\zeta_{t}(z)$ is bounded and supported in the unit ball, and $\partial_{r'}F_t(x-\eps z)dz$ is bounded, this shows that $\partial^{\bfalpha} F_{t,\eps}(x)$ is uniformly bounded. Moreover, if $|\bfalpha|=1,$ then the uniform bound is independent of $\eps.$

Now we claim that all of $\norm{\bfb^{t,\eps}-\bfb^t}_2, \norm{\bfu^{t,\eps}-\bfu^t}_2,\norm{\bfh^{t,\eps}-\bfh^t}_2$ and $\norm{\bfv^{t,\eps}-\bfv^t}_2$ are bounded by $\eps \sqrt{N}\cdot O_{t,L}(1) $ with high probability. This will be proved by induction, following the same induction order as the one in which the vectors are defined in \Cref{defn:GFOM}, i.e. $\bfv^0\to \bfb^1\to \bfu^1\to\bfh^1\to\bfv^1\to\bfb^2\to\cdots.$ Clearly as a base case we have $\norm{\bfv^{0,\eps} - \bfv^0}_2=0.$ For the inductive steps, we carry out $\bfb^t\to \bfu^t$ and $\bfu^t\to \bfh^t$, since the other inductive steps are analogous. For $\bfb^t\to \bfu^t,$ we assume that 
\[
\max_{s\in [t]}\norm{\bfb^{s,\eps} - \bfb^s}_2 \leq \eps \sqrt{N} O_{t,L}(1)
\]
(note that $\bfb^{*,\eps} = \bfb^*)$, and we aim to show that $\norm{\bfu^{t,\eps} - \bfu^t}_2\leq \eps \sqrt{N} O_{t,L}(1).$ We have
\begin{align*}
    \norm{\bfu^{t,\eps} - \bfu^t}_2 &= \norm{F_{t,\eps}(\bfb^*,\bfb^{1,\eps},\dots,\bfb^{t,\eps}) - F_t(\bfb^*,\bfb^{1},\dots,\bfb^{t})}_2 \\
    &\leq \norm{F_{t,\eps}(\bfb^*,\bfb^{1,\eps},\dots,\bfb^{t,\eps}) - F_{t,\eps}(\bfb^*,\bfb^{1},\dots,\bfb^{t})}_2  \\
    &\qquad + \norm{F_{t,\eps}(\bfb^*,\bfb^{1},\dots,\bfb^{t}) - F_{t}(\bfb^*,\bfb^{1},\dots,\bfb^{t})}_2 \\
    &\leq L O_{t,L}(1)\sum_{s=1}^t \norm{\bfb^{s,\eps} - \bfb^s}_2 + L\eps \sqrt{N} \\
    &\leq \eps\sqrt{N} O_{t,L}(1).
\end{align*}
Above, in the third step we have used that $\nabla F_{t,\eps}$ is uniformly bounded (independently of $\eps$) to bound the first term. This proves the inductive step $\bfb^t\to\bfu^t.$ For $\bfu^t\to\bfh^t,$ we assume that $\norm{\bfu^{t,\eps} - \bfu^t}_2 \leq \eps\sqrt{N} O_{t,L}(1)$ and we prove that $\norm{\bfh^{t,\eps} - \bfh^t}_2 \leq \eps\sqrt{N} O_{t,L}(1)$. Indeed, we have 
\begin{align*}
    \norm{\bfh^{t,\eps} - \bfh^t}_2 &= \norm{X^T\bfu^{t,\eps} - X^T\bfu^t}_2  \\
    &\leq \norm{X^T}_2\norm{\bfu^{t,\eps} - \bfu^t}_2.
\end{align*}
But we have $\norm{X^T}_2\leq O(1)$ with high probability (see e.g. \cite{vershynin2018high}). This proves that the inductive step $\bfu^t\to\bfh^t.$

Finally, let $\phi_{\eps},\psi_\eps$ be obtained from $\phi,\psi$ in exactly the same way that we obtained $F_{t,\eps},G_{t,\eps}$ from $F_t,G_t.$ To conclude the proof of the proposition, it suffices to show that we have 
\begin{align*}
    \frac{1}{d}\sum_{j=1}^d |\phi(\theta_j, \{h_j^t\}_{t=1}^T, \{v_j^t\}_{t=0}^T) - \phi_\eps(\theta_j, \{h_j^{t,\eps}\}_{t=1}^T, \{v_j^{t,\eps}\}_{t=0}^T)| &\leq \eps O_{T,L}(1) \\
    \frac{1}{N}\sum_{i=1}^N |\psi(b_i^*, \{b_i^t\}_{t=1}^T, \{u_i^t\}_{t=1}^T) - \psi_\eps(b_i^*, \{b_i^{t,\eps}\}_{t=1}^T, \{u_i^{t,\eps}\}_{t=1}^T)| &\leq \eps O_{T,L}(1).
\end{align*}
Since the proofs are analogous, we prove the first one. By the same reason as in \eqref{eq:F-eps-Lipschitz-bound} above, we have the same uniform bound for $\phi - \phi_\eps.$ Hence we have 
\begin{align*}
    \frac{1}{d}&\sum_{j=1}^d |\phi(\theta_j, \{h_j^t\}_{t=1}^T, \{v_j^t\}_{t=0}^T) - \phi_\eps(\theta_j, \{h_j^{t,\eps}\}_{t=1}^T, \{v_j^{t,\eps}\}_{t=0}^T)| \\
    &\leq \frac{1}{d}\sum_{j=1}^d |\phi(\theta_j, \{h_j^t\}_{t=1}^T, \{v_j^t\}_{t=0}^T) - \phi_\eps(\theta_j, \{h_j^t\}_{t=1}^T, \{v_j^t\}_{t=0}^T))| \\
    &\qquad \qquad + |\phi_\eps(\theta_j, \{h_j^t\}_{t=1}^T, \{v_j^t\}_{t=0}^T)) - \phi_\eps(\theta_j, \{h_j^{t,\eps}\}_{t=1}^T, \{v_j^{t,\eps}\}_{t=0}^T)| \\
    &\leq L\eps + \frac{1}{\sqrt{d}}\norm{\phi_\eps(\bftheta, \{\bfh^t\}_{t=1}^T, \{\bfv^t\}_{t=0}^T)) - \phi_\eps(\bftheta, \{\bfh^{t,\eps}\}_{t=1}^T, \{\bfv^{t,\eps}\}_{t=0}^T)}_2 \\
    &\leq L\eps + \frac{1}{\sqrt{d}}O_{T,L}(1)\sum_{t=1}^T \norm{\bfh^t - \bfh^{t,\eps}}_2 + \norm{\bfv^t - \bfv^{t,\eps}}_2 \\
    &\leq \eps O_{t,L}(1),
\end{align*}
as desired.
\end{proof}

\subsection{Derivative Bounds: Proof of \Cref{thm:derivative-bounds}}\label{sec:proof-of-derivative-bounds}

The proof will be by induction in the same order in which the different variables are defined in \Cref{defn:GFOM}, i.e., 
\begin{align}
    (\bfv^0,\bftheta,\bfb^*)\to \bfb^1\to \bfu^1\to \bfh^1 \to \bfv^1\to \bfb^2\to \cdots. \label{eq:induction-chain}
\end{align}
At each step of the induction, our aim will be to verify the bounds of \Cref{thm:derivative-bounds} for the current variable in the chain \eqref{eq:induction-chain}, assuming the validity of the bounds for all prior variables in the chain and all values of $p,m.$ 

The inductive cases for $\bfb^t$ and $\bfu^t$ are completely analogous to those of $\bfh^t$ and $\bfv^t$, respectively. Hence, for brevity we only explicitly carry out the arguments for $\bfb^t$ and $\bfu^t$. We will also use throughout various standard bounds on subgaussian and subexponential random variables which are stated in \Cref{sec:subgaussian-subexponential}.

\paragraph{Base case.} As a base case, we begin by checking the bounds of \Cref{thm:derivative-bounds} for $\bfv^0,\bftheta,\bfb^*.$ Let $P$ be a multi-set of variables as in \Cref{defn:P-partial-derivatives} such that $|P|=m.$ It's clear that, for $\bfv^0,\bftheta,$ we have
\begin{align*}
    \partial_{P}v^0_j &= \1\{|P|=0\}v^0_j + \1\{P = \{v^0_j\}\} \\
    \partial_{P}\theta_j &= \1\{|P|=0\}\theta_j + \1\{P = \{\theta_j\}\}.
\end{align*}
The right-hand sides have all moments $O(1)$ by \Cref{lem:subgaussian-subexp-moments}. Moreover, suppose $j$ were not all-reaching in $P$. This in particular would mean that $|P|>0, P \neq \{v^0_j\}$ and $P \neq \{\theta_j\}$, so the right-hand sides above would be identically zero. Hence, their moments are trivially $O(1/\sqrt{N}).$

Finally we check the base case of $b_i^*.$ We have 
\begin{align*}
    \partial_P b_i^* &= \1\{P = \{b_i^*\}\} + \partial_P \sum_{j=1}^d X_{ij}\theta_j \\
    &= \1\{P = \{b_i^*\}\} +\sum_{j=1}^d\1\{X_{ij}\in P\}\partial_{P\setminus X_{ij}}\theta_j + X_{ij}\partial_P \theta_j \\
    &= \1\{P = \{b_i^*\}\} +\sum_{j=1}^d \1\{P = \{X_{ij}\}\} \theta_j + \1\{|P|=0\}\sum_{j=1}^d X_{ij}\theta_j
\end{align*}
so 
\begin{align*}
    \norm{\partial_P b_i^* }_p &\leq \1\{P = \{b_i^*\}\} + \sum_{j=1}^d \1\{P = \{X_{ij}\}\} \norm{\theta_j}_p + \1\{|P|=0\}\norm{b_i^*}_p.
\end{align*}
The middle sum has at most one non-zero term, and hence is always $O(1)$ by \Cref{lem:subgaussian-subexp-moments}. For the last term, note that since $X_{ij}$ and $\theta_j$ are $\frac{\sigma}{\sqrt{N}}$-subgaussian and $\sigma$-subgaussian, respectively, their product is $\frac{\sigma^2}{\sqrt{N}}$-subexponetial by \Cref{lem:prod-of-subgaussians}. But then by \Cref{lem:sum-subgaussian-subexp}, their sum across $j$ is $\sqrt{C} \sigma^2$-subexponential, and by \Cref{lem:subgaussian-subexp-moments}, we get $\norm{b_i^*}_{p}\leq C^{3/2} \sigma^2 p = O(1).$ Hence we conclude 
\begin{align*}
    \norm{\partial_P b_i^* }_p &\leq O(1)
\end{align*}
always, checking part (1) of \Cref{thm:derivative-bounds}. For part (2), note that if $i$ is not all-reaching in $P,$ then we can neither have $P = \{b_i^*\}$ nor $|P|=0$ nor $P = \{X_{ij}\}$ for any $j$. Hence $\partial_P b_i^*$ is again equal to zero, and the bound is trivially satisfied.

\paragraph{Inductive case, $\bfb^t$.} We assume that the bounds of \Cref{thm:derivative-bounds} are true for all variables prior to $\bfb^{t}$ in the chain \eqref{eq:induction-chain} and aim to verify them for $b^{t}_i$ for some $i \in [N].$ Let $P$ be a multi-set of variables as in \Cref{defn:P-partial-derivatives} such that $|P|=m.$ Using that $b_i^{t} =\sum_{j=1}^d X_{ij}v_j^{t-1},$ we have
\begin{align}
    \norm{\partial_P b_i^t}_p &= \norm{\1\{P = \{b_i^t\}\} + \sum_{\substack{j\in [d]: X_{ij}\in P}} \partial_{P\setminus X_{ij}} v^{t-1}_j + \sum_{j=1}^d X_{ij}\partial_P v_j^{t-1}}_p  \nonumber\\
    &\leq \1\{P = \{b_i^t\}\} + \sum_{\substack{j\in [d]: X_{ij}\in P}} \norm{\partial_{P\setminus X_{ij}} v^{t-1}_j }_p+ \norm{\sum_{j=1}^d X_{ij}\partial_P v_j^{t-1}}_p \label{eq:expansion-h}
\end{align}
For the middle term, each term in the summand is $\leq \Gamma_{t,m-1,p}$ by inductive hypothesis, so this term is always $O(1).$ Moreover, suppose that $i$ was not all-reaching in $P$. Note that $P \neq \{b_i^t\}$, and $G(P\setminus X_{ij})$ is precisely the graph obtained from $G(P)$ by starting at $i$ and taking a step towards neighbor $j,$ deleting the edge we traveled through behind us. From this interpretation it's clear that $j$ is also not all reaching in $G(P\setminus X_{ij})$ and hence in this case every term in the summand is $\leq \Gamma_{t,m-1,p}/\sqrt{N}$. 

In summary, we have shown that the first two terms in \eqref{eq:expansion-h} are always $O(1),$ and are $O(1/\sqrt{N})$ if $i$ is not all-reaching in $P.$ What remains to conclude this inductive case is to check that the same is true for the last term in \eqref{eq:expansion-h}. We have
\begin{align}
    \norm{\sum_{j=1}^d X_{ij}\partial_P v_j^{t-1}}_p^p &= \sum_{j_1,\dots,j_p=1}^d \E X_{ij_1}\dots X_{ij_p} (\partial_Pv_{j_1}^{t-1}) \cdots (\partial_Pv_{j_p}^{t-1}) .\label{eq:h-pth-power}
\end{align}
   Now, given a tuple of indices $(j_1,\dots,j_p)$, let $U(j_1,\dots,j_p)\subseteq \{j_1,\dots,j_p\}$ be the subset of indices which appear in the tuple $(j_1,\dots,j_p)$ exactly once. For $0\leq q\leq p,$ we say that $(j_1,\dots,j_p) \in \calA_q$ if $|U(j_1,\dots,j_p)|=q,$ namely, the number of indices appearing exactly once in the tuple is exactly $q.$ Below, for $(j_1,\dots,j_p) \in \calA_q$, we will label the indices in $U(j_1,\dots,j_p)$ as $j'_1,\dots,j'_q$, and note that, by definition, the $j'_1,\dots,j'_q$ must all be distinct. Moreover, to avoid clutter we use the notation 
\begin{align}\label{eq:V-expansion}
    V_{j_1,\dots,j_p}(X) &:= (\partial_Pv_{j_1}^{t-1}) \cdots (\partial_Pv_{j_p}^{t-1}),
\end{align}
where the dependence on $X$ is made explicit. Now, we want to expand $V_{j_1,\dots,j_p}(X)$ around $X_{ij_1'}=\dots=X_{ij_q'}=0.$ To that end, for $\eta_1,\dots,\eta_q\in [0,1],$ define the the matrix $X^{i,\{j'_1,\dots,j'_q\}}(\eta_1,\dots,\eta_q) \in \R^{N\times d}$ as 
\begin{align}
    X^{i,\{j'_1,\dots,j'_q\}}(\eta_1,\dots,\eta_q)_{kl} = \begin{cases}
    \eta_\ell X_{i,j'_\ell}\qquad &\text{if }k =i\text{ and }l=j'_\ell \\
    X_{kl} & \text{otherwise.}
\end{cases}\label{eq:X-prime-expansion}
\end{align}
Then, by applying fundamental theorem of calculus $q$ times, we have 
\begin{align}
    V_{j_1,\dots,j_p}(X) &= V'_{j_1,\dots,j_p}(X) + \int_{[0,1]^q}  \partial_{X_{ij_1'}}\dots \partial_{X_{ij_p'}}V_{j_1,\dots,j_p}(X^{i,\{j'_1,\dots,j'_q\}}(\eta_1,\dots,\eta_q))\prod_{\ell=1}^q (X_{i,j'_\ell}d\eta_\ell), \label{eq:V-FTC}
\end{align}
where $V_{j_1,\dots,j_p}'(X)$ is made up of a sum of terms, each of which sets at least one of the variables $X_{i,j_1'},\dots,X_{i,j_q'}$ to zero. Hence, using the fact that the $X_{i,j'_\ell}$ have expectation zero and appear only once in $(j_1,\dots,j_p)$, we can simplify \eqref{eq:h-pth-power} as follows:
\begin{align}
    &\norm{\sum_{j=1}^d X_{ij}\partial_P v_j^{t-1}}_p^p \nonumber\\
    &= \sum_{q=0}^p \int_{[0,1]^q} \prod_{\ell=1}^q d\eta_\ell\sum_{(i_1,\dots,i_p) \in \calA_q} \E\left(\prod_{\ell: j_\ell \notin U(j_1,\dots,j_p)} X_{i,j_\ell}\right) \left(\prod_{\ell\in [q]} X_{i,j'_\ell}^2\right) \nonumber\\
    &\qquad \qquad \qquad \qquad \qquad \qquad \qquad\qquad \qquad\partial_{X_{ij'_1}}\dots \partial_{X_{ij'_q}} V_{j_1,\dots,j_p}(X^{i,\{j'_1,\dots,j'_q\}}(\eta_1,\dots,\eta_q))\nonumber \\
    &\lesssim  \sum_{q=0}^p \int_{[0,1]^q} \prod_{\ell=1}^q d\eta_\ell\frac{1}{N^{(p+q)/2}}\sum_{(j_1,\dots,j_p) \in \calA_q} \norm{\partial_{X_{ij'_1}}\dots \partial_{X_{ij'_q}} V_{j_1,\dots,j_p}(X^{i,\{j'_1,\dots,j'_q\}}(\eta_1,\dots,\eta_q))}_{p+1}, \label{eq:h-hard-term-expansion}
\end{align}
where in the last step we have used \Cref{lem:subgaussian-subexp-moments} and the general version of Hölder's inequality, and the notation $\lesssim$ means inequality up to a multiplicative constant independent of $N.$ We now make a claim that will be proved below.
\begin{claim}\label{claim:in-h-induction}
    Given a tuple of indices $(j_1,\dots,j_p) \in \calA_q$, define the two quantities
    \begin{align*}
        u(j_1,\dots,j_p) &:= |\{\ell \in [p] : j_\ell \in V(P)\}| \\
        u'(j_1,\dots,j_p) &:= |\{\ell \in [q] : j_\ell' \in V(P)\}|.
    \end{align*}
    Then, we always have 
    \begin{align}
        \norm{\partial_{X_{ij'_1}}\dots \partial_{X_{ij'_q}} V_{j_1,\dots,j_p}(X^{i,\{j'_1,\dots,j'_q\}}(\eta_1,\dots,\eta_q))}_{p+1} &\leq O(1), \label{eq:claim-in-h-induction-O1}
    \end{align}
    and if $i$ is not all-reaching in $P$, we also have 
    \begin{align}
        \norm{\partial_{X_{ij'_1}}\dots \partial_{X_{ij'_q}} V_{j_1,\dots,j_p}(X^{i,\{j'_1,\dots,j'_q\}}(\eta_1,\dots,\eta_q))}_{p+1} &\leq O\left(\left(\frac{1}{\sqrt{N}}\right)^{p - u(j_1,\dots,j_p) - u'(j_1,\dots,j_p)}\right).\label{eq:claim-in-h-induction-OsqrtN}
    \end{align}
\end{claim}
Before proving \Cref{claim:in-h-induction}, let us conclude the inductive case for $b_i^t$ from it and \eqref{eq:h-hard-term-expansion}. Recall that it suffices for us to bound \eqref{eq:h-hard-term-expansion} by $O(1)$ in general and by $O( N^{-p/2})$ in case that $i$ is not all-reaching in $P$. For the first case, note that each tuple in $\calA$ has at most $q + (p-q)/2 = (p+q)/2$ distinct indices, and hence $|\calA_q|\lesssim N^{(p+q)/2}.$ This and the first case of \Cref{claim:in-h-induction} lets us conclude that
\begin{align*}
    \norm{\sum_{j=1}^d X_{ij}\partial_P v_j^{t-1}}_p^p &\leq O(1)
\end{align*}
in general, as desired. Now suppose that $i$ is not all-reaching in $P$. Then the second case of \Cref{claim:in-h-induction} and \eqref{eq:h-hard-term-expansion} imply that
\begin{align*}
    \norm{\sum_{j=1}^d X_{ij}\partial_P v_j^{t-1}}_p^p &\lesssim \sum_{q=0}^p \frac{1}{N^{(p+q)/2}}\sum_{(j_1,\dots,j_p) \in \calA_q} \left(\frac{1}{\sqrt{N}}\right)^{p-u(j_1,\dots,j_p)-u'(j_1,\dots,j_p)} \\
    &\lesssim \sum_{q=0}^p \Eop_{(j_1,\dots,j_p) \in \calA_q} \left(\frac{1}{\sqrt{N}}\right)^{p-u-u'} \\
    &= N^{-p/2} \sum_{q=0}^p \Eop_{(j_1,\dots,j_p) \in \calA_q} N^{u'}N^{(u-u')/2}
\end{align*}
where we have used the shorthand notation $u = u(j_1,\dots,j_p),u'=u'(j_1,\dots,j_p),$ so to conclude it suffices to show that 
\[
\Eop_{(j_1,\dots,j_p) \in \calA_q} N^{u'}N^{(u-u')/2}\leq O(1).
\]
Note that $(j_1,\dots,j_p)\sim \calA_q$ may be sampled in three steps. First, we sample distinct $j_1',\dots,j_q'$ uniformly at random in $[d].$ Then, we sample $j''_1,\dots,j''_{p-q}$ uniformly at random from $[d]\setminus\{j'_1,\dots,j'_q\}$ conditional on the fact that no $j''_k$ appears uniquely in $j''_1,\dots,j''_{p-q}$. Finally, we let $(j_1,\dots,j_p)$ be the union of $j'_1,\dots,j'_q$ and $j''_1,\dots,j''_{p-q}$ arranged in a random order. Then $u'$ is the number of $j'_k$ that lie in $V(P)$ and $u-u'$ is the number of $j''_k$ that lie in $V(P).$ Hence we have
\begin{align*}
    \Eop_{j''_1,\dots,j''_{p-q}}\left[N^{(u-u')/2}\; \Big|\;\{j_1',\dots,j'_q\}\right] &=  \frac{\sum_{s=0}^{(p-q)/2}\sum_{S\in \binom{[d]\setminus \{j_1',\dots,j_q'\}}{s}} \sum_{\substack{\{j''_1,\dots,j''_{p-q} \} =  S: \\ \text{no uniques}}} N^{(u-u')/2} }
    {\sum_{s=0}^{(p-q)/2}\sum_{S\in \binom{[d]\setminus \{j_1',\dots,j_q'\}}{s}} \sum_{\substack{\{j''_1,\dots,j''_{p-q} \} =  S: \\ \text{no uniques}}} 1 } \\
    &\leq 
    \frac{\sum_{s=0}^{(p-q)/2}s^{p-q}\sum_{S\in \binom{[d]\setminus \{j_1',\dots,j_q'\}}{s}}  d^{|S\cap P| + \1_{S\cap P \neq \emptyset}(p-q-2s)/2} }
    {\sum_{s=0}^{(p-q)/2} |\binom{[d]\setminus \{j_1',\dots,j_q'\}}{s}| }.
\end{align*}
Now we have 
\begin{align*}
    \sum_{S\in \binom{[d]\setminus \{j_1',\dots,j_q'\}}{s}}&  d^{|S\cap P| + \1_{S\cap P \neq \emptyset}(p-q-2s)/2} \\
    &\leq \left|\binom{[d]\setminus \{j_1',\dots,j_q'\}}{s}\right| \left(1+\sum_{k=1}^s \Pr_{S\sim \binom{[d]\setminus \{j_1',\dots,j_q'\}}{s}}[|S\cap P|\geq k] d^{k+ (p-q)/2- s} \right) \\
    &\leq \left|\binom{[d]\setminus \{j_1',\dots,j_q'\}}{s}\right| \left(1+\sum_{k=1}^s 2|P|^k d^{ (p-q)/2- s} \right) \\
    &\leq 4 d^{ (p-q)/2- s}\left|\binom{[d]\setminus \{j_1',\dots,j_q'\}}{s}\right| \sum_{k=1}^s |P|^k  \\
    &\lesssim   N^{ (p-q)/2}.
\end{align*}
Thus we get
\begin{align*}
    \Eop_{j''_1,\dots,j''_{p-q}}\left[d^{(u-u')/2}\; \Big|\;\{j_1',\dots,j'_q\}\right] &\lesssim \frac{d^{(p-q)/2}}{\sum_{s=0}^{(p-q)/2} |\binom{[d]\setminus \{j_1',\dots,j_q'\}}{s}| } \\
    &\lesssim 1.
\end{align*}
Hence 
\begin{align*}
    \Eop_{(j_1,\dots,j_p) \in \calA_q} N^{u'}N^{(u-u')/2} &\lesssim \Eop_{(j_1,\dots,j_p) \in \calA_q} \left[N^{u'}\Eop_{j''_1,\dots,j''_{p-q}}\left[N^{(u-u')/2}\; \Big|\;\{j_1',\dots,j'_q\}\right]\right] \\
    &\lesssim \Eop_{j_1',\dots,j_q'} \left[N^{u'}\right] \\
    &\lesssim 1,
\end{align*}
where in the last step we have used that $j_1',\dots,j_q'$ is a uniformly-random set of size $q$ in $[d]$, and $u'$ is the number of the $j_1',\dots,j_q'$ that lie in a subset of $[d]$ of size $ O(1).$ This concludes the proof of this inductive step of $\bfb^t$ modulo \Cref{claim:in-h-induction}. Below, we prove this claim and hence conclude the proof of the inductive step of $\bfb^t$.

\begin{proof}[of \Cref{claim:in-h-induction}]
Let $Q = \{X_{ij_1'},\dots,X_{ij'_q}\}.$ We have
\begin{align*}
    &\partial_Q V_{j_1,\dots,j_p}(X^{i,\{j'_1,\dots,j'_q\}}(\eta_1,\dots,\eta_q)) \\
    &=  \sum_{Q_1\sqcup\dots\sqcup Q_p = Q} \prod_{k=1}^p (\partial_{Q_k}\partial_P v^{t-1}_{j_k})(X^{i,\{j'_1,\dots,j'_q\}}(\eta_1,\dots,\eta_q))
\end{align*}
where the sum is over all partitions of $Q$ into $p$ sets, and the sets in the partition are allowed to be empty. Hence, by the triangle inequality and Hölder inequality, we have
\begin{align*}
    &\norm{\partial_Q V_{j_1,\dots,j_p}(X^{i,\{j'_1,\dots,j'_q\}}(\eta_1,\dots,\eta_q))}_{p+1} \\
    &\qquad\qquad \qquad\qquad\qquad=  \max_{Q_1\sqcup\dots\sqcup Q_p = Q} \prod_{k=1}^p\norm{(\partial_{Q_k}\partial_P v^{t-1}_{j_k})(X^{i,\{j'_1,\dots,j'_q\}}(\eta_1,\dots,\eta_q))}_{p(p+1)}
\end{align*}
To conclude the proof, we will show that every term in this max satisfies \eqref{eq:claim-in-h-induction-O1}, and if $i$ is not all-reaching in $P,$ also \eqref{eq:claim-in-h-induction-OsqrtN}. For the first part, note that, if all the coordinates of $X$ are $\sigma/\sqrt{N}$-subgaussian, so are all the coordinates of $X^{i,\{j'_1,\dots,j'_q\}}(\eta_1,\dots,\eta_q),$ since the only possible difference between the two is that some coordinates are scaled down. Hence, by inductive hypothesis, every term is $O(1)$, as desired.

    Finally, suppose $i$ is not all reaching in $P,$ and let $\ell \in V(P)$ be such that there is no path from $i$ to $\ell$ in $G(P).$ To conclude the proof, we will show that, for every element in the max above, at least $p-u(j_1,\dots,j_p)-u'(j_1,\dots,j_p)$ of the $p$ terms in the product are such that $\ell$ is also not reachable from $j_k$ in $G(Q_k\cup P).$ To see this, note that there are at least $p - u(j_1,\dots,j_p) - u'(j_1,\dots,j_p)$ terms $j_k$ such that $j_k \notin V(P)$ and $Q_k$ only connects vertices outside $V(P)$ to $i$. But now, since $j_k \notin V(P)$ and $Q_k$ consists only of connections from points outside $V(P)$ to $i$, any path going from $j_k$ to $\ell$ in $G(Q_k\cup P)$ must first travel to $i$, and then from $i$ to $\ell$ \emph{within} $G(P),$ which is a contradiction. This concludes the proof.
\end{proof}

\paragraph{Inductive case, $\bfu^t.$} To address this case, we first prove a general lemma. 
\begin{lemma}\label{lem:chain-rule-expansion}
    Let $f:\R^k \to\R$ be a smooth Lipschitz function, with all derivatives of order $k$ also Lipschitz with Lipschitz constant $L_k,$ for $k=0,1,\dots.$ Let $x_1,\dots,x_n$ be random variables, and for $\ell=1,\dots,k$, let $a_\ell = a_\ell(x_1,\dots,x_n)$ be smooth functions of these random variables. Then for every tuple $P = (i_1,\dots,i_m) \in [n]^m$, define
    \[
    \partial_P := \partial_{x_{i_1}}\dots\partial_{x_{i_m}}.
    \]
    Then, for all $p\geq 1,$ we have
    \begin{align*}
        \norm{\partial_P f(a_1,\dots,a_k)}_p &\leq \sum_{r=1}^m L_r \cdot k^{r+1}\cdot r^m\cdot \max_{j_0,\dots,j_r\in [k]} \norm{a_{j_0}}_{p(r+1)}\max_{P_1\sqcup\cdots\sqcup P_r=P}\prod_{\ell=1}^r\norm{\partial_{P_\ell} a_{j_\ell}}_{p(r+1)}.
    \end{align*}
\end{lemma}
Before proving the lemma, note that the inductive case for $\bfu^t$ follows from it. Indeed, since 
\begin{align*}
u_i^t &= F_t(b_i^*, b_i^1,\dots,b_i^t) 
\end{align*}
and by inductive assumption all moments of partial derivatives of the inputs to $F_t$ above are $O(1)$, the same holds for all moments of partial derivatives of $u^t_i$ by \Cref{lem:chain-rule-expansion}. Moreover, suppose $i$ is not all reaching in $P$. Note that, for every partition $P_1\sqcup\cdots\sqcup P_r$ of $P,$ there must exist some $\ell$ such that $i$ is not all-reaching in $P_\ell.$ Hence, in \Cref{lem:chain-rule-expansion}, by the inductive hypothesis, all terms in the max have a term in the product that is $O( 1/\sqrt{N}),$ and hence so will be the max itself. We conclude the inductive case for $\bfu^t$ by proving \Cref{lem:chain-rule-expansion}.

\begin{proof}[of \Cref{lem:chain-rule-expansion}]
    It's shown in \cite[Lemma 5.1]{wei-kwo-universality} that we have 
    \begin{align*}
        \partial_P f(a_1,\dots,a_k) &= \sum_{r=1}^m \sum_{1\leq j_1,\dots,j_r\leq k} \sum_{P_1\sqcup \cdots \sqcup P_r = P} (\partial_{j_1}\dots\partial_{j_r}f)(a_1,\dots,a_k) (\partial_{P_1}a_{j_1}) \dots (\partial_{P_r}a_{j_r}),
    \end{align*}
    so
    \begin{align*}
        &\norm{\partial_P f(a_1,\dots,a_k)}_p \\
        &\leq \sum_{r=1}^m \sum_{1\leq j_1,\dots,j_r\leq k} \sum_{P_1\sqcup \cdots \sqcup P_r = P} \norm{(\partial_{j_1}\dots\partial_{j_r}f)(a_1,\dots,a_k)}_{p(r+1)}\norm{\partial_{P_1}a_{j_1}}_{p(r+1)} \dots \norm{\partial_{P_r}a_{j_r}}_{p(r+1)} \\
        &\leq \sum_{r=1}^m \sum_{1\leq j_1,\dots,j_r\leq k} \sum_{P_1\sqcup \cdots \sqcup P_r = P} L_r\left(\norm{a_1}_{p(r+1)} + \dots + \norm{a_k}_{p(r+1)} \right)\norm{\partial_{P_1}a_{j_1}}_{p(r+1)} \dots \norm{\partial_{P_r}a_{j_r}}_{p(r+1)} \\
        &\leq \sum_{r=1}^m L_r \cdot k^{r+1}\cdot r^m\cdot \max_{j_0,\dots,j_r\in [k]} \norm{a_{j_0}}_{p(r+1)}\max_{P_1\sqcup\cdots\sqcup P_r=P}\prod_{\ell=1}^r\norm{\partial_{P_\ell} a_{j_\ell}}_{p(r+1)},
    \end{align*}
    as desired.
\end{proof}

\subsection{Proof of Main Result}\label{sec:main-proof}
In this section we prove \Cref{thm:main-strong} and \Cref{thm:main}. We begin by showing that the former implies the latter.

\subsubsection{Reduction to \Cref{thm:main-strong}}

\begin{lemma}\label{lem:main-strong-implies-main}
    \Cref{thm:main-strong} implies \Cref{thm:main}.
\end{lemma}
\begin{proof}
    Suppose \Cref{thm:main-strong} holds. By \Cref{prop:smooth-approximation}, it suffices to prove \Cref{thm:main} in the special case where $\phi,\psi,F_t,G_t$ are all $C^\infty_{1b}.$ But part (1) of \Cref{thm:main-strong} implies that 
    \[
    \E \int \phi d\nu_{N,T} \to \int \phi d\nu_{T},
    \]
    and part (2) implies that 
    \[
    \Var\left(\int \phi d\nu_{N,T} \right)\to 0.
    \]
    Chebyshev's inequality then implies that 
    \[
    \int \phi d\nu_{N,T} \to \int \phi d\nu_{T}
    \]
    in probability, and similarly
    \[
        \int \psi d\mu_{N,T} \to \int \psi d\mu_{T}
    \]
    in probability. 
\end{proof}

For the rest of the section, we prove \Cref{thm:main-strong}.
\subsubsection{Reduction to Fields}
We begin by noting that, to prove \Cref{thm:main-strong}, it suffices to consider functions of the field variables only.
\begin{theorem}\label{thm:reduction-to-fields}
    Let $\overline{\mu}_T, \overline{\nu}_T$ be the extensions of the measures defined by the DMFT equations \eqref{eq:mu-0-base-case}-\eqref{eq:nu-T-recursion} defined as follows:
    \begin{align*}
        \overline{\mu}_T &:= \Law(b^*, \{b^t\}_{t=1}^T, \{u^t\}_{t=1}^T, \{w^t\}_{t=1}^T) \\
        \overline{\nu}_T &:= \Law(\theta, \{h^t\}_{t=1}^T, \{v^t\}_{t=0}^T, \{\Tilde{w}^t\}_{t=1}^T).
    \end{align*}
    Note that $\mu_T$ and $\nu_T$ can be obtained from $\overline{\mu}_T$ and $\overline{\nu}_T$  by marginalizing out the noise variables $\{w^t\}_{t=1}^T$ and $\{\tilde{w}^t\}_{t=1}^T$, respectively. Under the same assumptions as \Cref{thm:main-strong}, for every pair of $C^\infty_{1b}$ functions $\phi:\R^{2T+2}\to\R$ and $\psi:\R^{2T+1}\to \R$, the following hold:
    \begin{enumerate}
        \item For all $i\in [N]$ and $j\in [d],$
        \begin{align}
            \lim |\E \phi( \{h_j^t\}_{t=1}^T;v^0_j,\theta_j,\{\tilde{w}_j^t\}_{t=1}^T) - \E_{\overline{\nu}_T} \phi( \{h^t\}_{t=1}^T;v^0,\theta,\{\tilde{w}^t\}_{t=1}^T)| &= 0 \label{eq:nu-marginal-covergence} \\
            \lim |\E \psi(b_i^*, \{b_i^t\}_{t=1}^T; \{w^t_i\}_{t=1}^T) - \E_{\overline{\mu}_T} \psi(b^*, \{b^t\}_{t=1}^T;\{w^t\}_{t=1}^T)| &= 0.\label{eq:mu-marginal-covergence} 
        \end{align}
        \item For all $i\neq i' \in [N]$ and $j\neq j'\in [d]$, we have
        \begin{align}
            \lim  \left|\Cov\left(\phi( \{h_j^t\}_{t=1}^T;v^0_j,\theta_j,\{\tilde{w}_j^t\}_{t=1}^T),\; \phi( \{h_{j'}^t\}_{t=1}^T;v^0_{j'},\theta_{j'},\{\tilde{w}_{j'}^t\}_{t=1}^T)\right)\right| &= 0 \label{eq:nu-decoupling}\\
            \lim  \left|\Cov\left( \psi(b_i^*, \{b_i^t\}_{t=1}^T; \{w^t_i\}_{t=1}^T),\; \psi(b_{i'}^*, \{b_{i'}^t\}_{t=1}^T; \{w^t_{i'}\}_{t=1}^T)\right)\right| &= 0.\label{eq:mu-decoupling}
        \end{align}
        \end{enumerate}
\end{theorem}
\begin{lemma}
    \Cref{thm:reduction-to-fields} implies \Cref{thm:main-strong}.
\end{lemma}
\begin{proof}
    This follows from the fact that we have $u_i^t =F_t(b^*_i,b^1_i,\dots,b_i^t; w_i^1,\dots,w_i^t)$ and $v_j^t = G_t(h^1_j,\dots,h^t_j;v^0_j,\tilde{w}^1_j,\dots,\tilde{w}^t_j)$, and we can absorb all functions $F_t,G_t$ into $\phi,\psi.$ The composition of two $C^\infty_{1b}$ functions is also $C^\infty_{1b},$ so the lemma follows.
\end{proof}

\subsubsection{Induction Outline}\label{sec:induction-outline}
The proof of \Cref{thm:reduction-to-fields} will be by induction. The induction order will be, as usual, the same as the order of definition in \Cref{defn:GFOM}:
{\small
\begin{align*}
    \{\text{\eqref{eq:nu-marginal-covergence}-\eqref{eq:mu-decoupling} at }T=0 \}\to \{\eqref{eq:mu-marginal-covergence},\eqref{eq:mu-decoupling} \text{ at }T=1 \}\to \{\eqref{eq:nu-marginal-covergence},\eqref{eq:nu-decoupling} \text{ at }T=1 \}\to \{\eqref{eq:mu-marginal-covergence},\eqref{eq:mu-decoupling} \text{ at }T=2 \}\to \cdots.
\end{align*}
}
At each step in the chain above, we will prove the corresponding statement while assuming all prior statements in the chain. There are two types of inductive statements: proving \eqref{eq:nu-marginal-covergence}, \eqref{eq:nu-decoupling} at step $T,$ and proving \eqref{eq:mu-marginal-covergence}, \eqref{eq:mu-decoupling} at step $T.$ Since the proofs of these two types of inductive cases are completely analogous, we will carry out only the former explicitly. 

\subsubsection{Base Case}
For $T=0,$ equations \eqref{eq:nu-marginal-covergence} and \eqref{eq:nu-decoupling} are immediate since $(v_j^0,\theta_j^0)$ is equal in distribution to $v^0,\theta$ under $\nu_0$ and $(v_j^0,\theta_j^0), (v_{j'}^0,\theta_{j'}^0)$ for $j\neq j'$ are independent.

To prove equations \eqref{eq:mu-marginal-covergence} and \eqref{eq:mu-decoupling}, we must show that for every $C^\infty_{1b}$ function $\psi:\R\to\R,$ and for all $i\neq i'\in [N]$, we have 
    \begin{align}
         |\E\psi(b^*_i) - \E \psi(b^*)|\to 0,\label{eq:main-thm-base-case-1} \\
        \max_{i\neq i' \in [N]} |\E\psi(b^*_i)\psi(b^*_{i'}) - \E\psi(b^*_i)\E\psi(b^*_{i'}) | \to 0,\label{eq:main-thm-base-case-2}
    \end{align}
    where $b\sim N(0,\alpha \E_{\nu_0}[\theta^2])$. Now recall that we have
    \begin{align*}
        b^*_i &= \sum_{j=1}^d X_{ij} \theta_j 
    \end{align*}
    where the $X_{ij}$ are independent and independent of the $\theta_j.$ Hence \eqref{eq:main-thm-base-case-1} follows from the classical Berry–Esseen Theorem. The case of \eqref{eq:main-thm-base-case-2} is similar, but we carry it out explicitly for completeness. Let $\{\theta'_j\}_{j=1}^N$ be i.i.d. samples equal in law to $\theta_1$, but independent from $\bftheta.$ For $0\leq k\leq d$, define
    \begin{align*}
        S_{i'}^k &= \sum_{j=1}^k X_{i'j} \theta_j + \sum_{j=k+1}^d X_{i'j} \theta_j' \\
        S_{i'}^{-k} &= \sum_{j=1}^{k-1} X_{i'j} \theta_j + \sum_{j=k+1}^d X_{i'j} \theta_j' 
    \end{align*}
    so that $S^d_{i'} = b^*_{i'}.$ Note that \eqref{eq:main-thm-base-case-2} can be equivalently written as 
    \begin{align*}
        \max_{i\neq i' \in [N]} |\E\psi(b^*_i)\psi(S_{i'}^d) - \E\psi(b^*_i)\psi(S_{i'}^0) | \to 0.
    \end{align*}
    Now we have 
    \begin{align*}
        &|\E\psi(b^*_i)\psi(S_{i'}^d) - \E\psi(b^*_i)\psi(S_{i'}^0)| \leq \sum_{k=1}^d |\E\psi(b^*_i)(\psi(S_{i'}^k) - \psi(S_{i'}^{k-1})) | \\
        &\leq \sum_{k=1}^d |\E\psi(b^*_i)\psi'(S_{i'}^{-k})(X_{i'j}\theta_j - X_{i'j}\theta_j')| + \frac{1}{2} |\E\psi(b^*_i)\psi''(S_{i'}^{-k})((X_{i'j}\theta_j)^2 - (X_{i'j}\theta_j')^2)| \\
        &\qquad \qquad \qquad \qquad \qquad +\frac{\sup_x \psi'''(x)}{6}(\E|\psi(b^*_i)(X_{i'j}\theta_j)^3| + \E|\psi(b^*_i)(X_{i'j}\theta_j')^3|).
    \end{align*}
    The third term inside the outer sum is clearly $ O(N^{-3/2})$, so it suffices to show that the first two are as well. Note that the first term is equal to zero because neither $b^*_i$ nor $S^{-k}_{i'}$ depend on $X_{i'j}.$ Finally, for the second term we have
    \begin{align*}
        &\E\psi(b^*_i)\psi''(S_{i'}^{-k})((X_{i'j}\theta_j)^2 - (X_{i'j}\theta_j')^2) = \frac{1}{N}\left(\E\psi(b^*_i)\psi''(S_{i'}^{-k})\theta_j^2 - \E\psi(b^*_i)\psi''(S_{i'}^{-k})\E\theta_j^2\right) \\
         &\qquad \qquad\lesssim \frac{1}{N}\left(\E\psi(b^*_i-X_{ij}\theta_j)\psi''(S_{i'}^{-k})\theta_j^2 + \E|X_{ij}\theta_j\psi''(S_{i'}^{-k})\theta_j^2| - \E\psi(b^*_i)\psi''(S_{i'}^{-k})\E\theta_j^2\right) \\
        &\qquad \qquad\lesssim \frac{1}{N}\left(\E\psi(b^*_i-X_{ij}\theta_j)\psi''(S_{i'}^{-k})\E\theta_j^2 - \E\psi(b^*_i)\psi''(S_{i'}^{-k})\E\theta_j^2 + \frac{1}{\sqrt{N}}\right) \\
        &\qquad \qquad\lesssim N^{-3/2}
    \end{align*}
    where in the second and last steps we have used the Lipschitzness of $\psi$, and in the third step we have used that $b^*_i - X_{ij}\theta_j$ is independent of $\theta_j$. This concludes the base case.

\subsubsection{Inductive Case}
Here we do the inductive case. As explained in \Cref{sec:induction-outline}, we only prove equations \eqref{eq:nu-marginal-covergence}, \eqref{eq:nu-decoupling} at time $T$. For the rest of \Cref{sec:main-proof}, we assume all previous links in the chain in \Cref{sec:induction-outline} hold.

Our proof of \eqref{eq:nu-marginal-covergence}, \eqref{eq:nu-decoupling} at time $T$ will be split into two lemmas. To state the lemmas, we use the notation
    \begin{align*}
        g_j^t &= \sum_{i=1}^N X_{ij}u_i^t(X^{ij})
    \end{align*}
for $t=1,\dots,T,$ where recall that $X^{ij}$ is equal to $X$ on all entries except for $X^{ij}_{ij}=0.$
\begin{lemma}\label{lem:inductive-case-lemma1}
    There exists a $C^\infty_{1b}$ function $\tilde{\phi}$ such that, for every $j\in [d]$,
    \begin{align}
        \lim \norm{\tilde{\phi}( \{g_j^t\}_{t=1}^T;v^0_j,\theta_j,\{\tilde{w}_j^t\}_{t=1}^T) - \phi( \{h_j^t\}_{t=1}^T;v^0_j,\theta_j,\{\tilde{w}_j^t\}_{t=1}^T) }_2 &= 0.\label{eq:gjt-mu-t-close}
    \end{align}
    Moreover, we have 
    \begin{align}
        \E \tilde{\phi}( \{g^t\}_{t=1}^T;v^0,\theta,\{\tilde{w}^t\}_{t=1}^T) &=\E_{\overline{\nu}_T} \phi( \{h^t\}_{t=1}^T;v^0,\theta,\{\tilde{w}^t\}_{t=1}^T),\label{eq:gt-mu-t-equal}
    \end{align}
    where on the right-hand side, $\{g^t\}_{t=1}^T$ is a Gaussian vector with covariance given by \eqref{eq:mu-T-covariance}, which is independent of $(v^0,\theta)\sim \nu_0$ and $\{\tilde{w}^t\}_{t=1}^T\sim N(0, I_T).$
\end{lemma}
\begin{lemma}\label{lem:inductive-case-lemma2}
    Let $\tilde{\psi}$, and $\{g^t\}_{t=1}^T$ be as in \Cref{lem:inductive-case-lemma1}.
    For all $j\in [d]$, we have 
        \begin{align}
            \lim |\E \tilde{\phi}( \{g_j^t\}_{t=1}^T;v^0_j,\theta_j,\{\tilde{w}_j^t\}_{t=1}^T) - \E \tilde{\phi}( \{g^t\}_{t=1}^T;v^0,\theta,\{\tilde{w}^t\}_{t=1}^T)| &= 0 \label{eq:gt-marginal-convergence} 
        \end{align}
        and for all $j\neq j'\in [d]$, we have
        \begin{align}
            \lim  \left|\Cov\left(\tilde{\phi}( \{g_j^t\}_{t=1}^T;v^0_j,\theta_j,\{\tilde{w}_j^t\}_{t=1}^T) ,\; \tilde{\phi}( \{g_{j'}^t\}_{t=1}^T;v^0_{j'},\theta_{j'},\{\tilde{w}_{j'}^t\}_{t=1}^T) \right)\right| &= 0. \label{eq:gt-decoupling}
        \end{align}
\end{lemma}
Note that \eqref{eq:nu-marginal-covergence}, \eqref{eq:nu-decoupling} at time $T$ follow from them. Indeed, \eqref{eq:gjt-mu-t-close} and \eqref{eq:gt-mu-t-equal} show that \eqref{eq:gt-marginal-convergence} and \eqref{eq:nu-marginal-covergence} are equivalent. Similarly, \eqref{eq:gjt-mu-t-close} shows that \eqref{eq:gt-decoupling} and \eqref{eq:nu-decoupling} are equivalent.

To conclude the proof of \Cref{thm:main-strong}, it remains to prove \Cref{lem:inductive-case-lemma1} and \Cref{lem:inductive-case-lemma2}. We do this in the next two subsections.

\subsubsection{Proof of \Cref{lem:inductive-case-lemma1}}
    Using the same expansion as in \eqref{eq:b-i-t-expansion}, we have 
    \begin{align*}
        h_j^t &= \sum_{i=1}^N X_{ij}u_i^{t}(X)\nonumber \\
    &= \sum_{i=1}^N X_{ij}u_i^{t}(X^{ij}) +  \sum_{i=1}^N X_{ij}^2\partial_{X_{ij}}u_i^{t}(X^{ij}) +\frac{1}{2} \int_0^1 d\eta(1-\eta) \sum_{i=1}^N X_{ij}^3 \partial^2_{X_{ij}}u_i^{t}(X^{ij}(\eta)).
\end{align*}
Now by chain rule, we have 
\begin{align*}
    \partial_{X_{ij}}u_i^{t}  (X^{ij}) &= \theta_j\partial_{b_i^*}u_i^{t}  (X^{ij}) + \sum_{s=1}^t v^{s-1}_j(X^{ij})\partial_{b_i^s}u_i^{t}  (X^{ij})  + \sum_{s=1}^{t-1} u^{s}_i(X^{ij})\partial_{h_j^s}u_i^{t}  (X^{ij})
\end{align*}
so we can write 
\begin{align}
    h_j^t &= g_j^t + \theta_j\sum_{i=1}^N X_{ij}^2\partial_{b_i^*}u_i^{t}  (X^{ij}) + \sum_{s=1}^t v^{s-1}_j(X^{ij})\sum_{i=1}^N X_{ij}^2\partial_{b_i^s}u_i^{t}  (X^{ij}) + \Delta^t_j\label{eq:h-t-j-expansion} \\
    \Delta^t_j &:= \sum_{s=1}^{t-1} u^{s}_i(X^{ij})\sum_{i=1}^N X_{ij}^2\partial_{h_j^s}u_i^{t}  (X^{ij}) + \frac{1}{2} \int_0^1 d\eta(1-\eta) \sum_{i=1}^N X_{ij}^3 \partial^2_{X_{ij}}u_i^{t}(X^{ij}(\eta)).\nonumber
\end{align}

The proof of \Cref{lem:inductive-case-lemma1} will be split in two parts, which will gradually simplify \eqref{eq:h-t-j-expansion}. In the first part, we show that the partial derivatives $\partial_{b_i^*}u_i^{t}  (X^{ij})$ and $\partial_{b_i^s}u_i^{t}  (X^{ij})$ can be approximated by operators that are defined in the same way that $\partial_{b^*}$ and $\partial_{b^s}$ are defined in \Cref{rem:meaning-of-field-derivatives-in-dmft}. This amounts to showing that $\partial_{b_i^*}u_i^{t}  (X^{ij}), \partial_{b_i^s}u_i^{t}  (X^{ij})$ can be approximated by considering only the \emph{explicit} influence of the fields $b_i^*, b_i^s$ on $u_i^t$, since the implicit influences through other dynamical variables are lower-order. In the second part, we show that the terms $\sum_{i=1}^N X_{ij}^2\partial_{b_i^*}u_i^{t}  (X^{ij})$ and $\sum_{i=1}^N X_{ij}^2\partial_{b_i^s}u_i^{t}  (X^{ij})$ concentrate. 

For the rest of the section, to avoid clutter, we use the notation $u_{i/ij}^t:= u_{i}^t(X^{ij}),$ and similarly e.g. $\partial_{b^*_i}u_{i/ij}^t:= \partial_{b^*_i}u_{i}^t(X^{ij})$ and $v_{j/ij}^{s-1}:= v_{j}^{s-1}(X^{ij}).$

\paragraph{Simplification of the partial derivatives.} We will define a new set of variables which approximate those of \Cref{defn:GFOM}, and are denoted by a tilde. We set $\tilde{\bftheta}=\bftheta,\tilde{\bfv}^0=\bfv^0,\tilde{\bfb}^* = \bfb^*.$ Moreover, for $t=1,\dots,T,$ we let
\[
z_j^t := \sum_{j=1}^d X_{ij}v_{j/ij}^{t-1} 
\]
and 
\begin{align*}
    \begin{cases}
        \tilde{b}_i^t = z_j^t + \sum_{s=1}^{t-1} \tilde{u}_i^s \sum_{j=1}^d X_{ij}^2 \tilde{\partial}_{\tilde{h}^s_j} \tilde{v}_j^{t-1} &\forall i\in [N]\\
        \tilde{u}_i^t = F_t(\tilde{b}_i^*, \tilde{b}_i^1,\dots,\tilde{b}_i^t;w^1_i,\dots,w^t_i) &\forall i \in [N] \\
        \tilde{h}_j^t = g_j^t + \tilde{\theta}_j \sum_{i=1}^N X_{ij}^2 \tilde{\partial}_{\tilde{b}_i^*} \tilde{u}_i^t + \sum_{s=1}^t \tilde{v}_j^{s-1}\sum_{i=1}^N X_{ij}^2 \tilde{\partial}_{\tilde{b}_i^s} \tilde{u}_i^t & \forall j\in [d] \\
        \tilde{v}_j^t = G_t(\tilde{h}_j^1,\dots,\tilde{h}_j^t;\tilde{v}_j^0,\tilde{w}^1_j,\dots,\tilde{w}^t_j) & \forall j\in [d].
    \end{cases}
\end{align*}
Above, the action of the operators $\tilde{\partial}$ are defined exactly the same as in \Cref{rem:meaning-of-field-derivatives-in-dmft}. Explicitly, we can define 
\begin{align*}
    \tilde{\partial}_{\tilde{b}_i^s}\tilde{b}_i^t &= \begin{cases}
            0 &\text{if }s>t \\ 
            1 \qquad &\text{if }s=t \\
            \sum_{r=1}^{t-1}\tilde{\partial}_{\tilde{b}_i^s} \tilde{u}_i^{r}\sum_{j=1}^d X_{ij}^2 \tilde{\partial}_{\tilde{h}_j^r} \tilde{v}_{j}^{t-1} & \text{if }s<t,
        \end{cases} \\
        \tilde{\partial}_{\tilde{b}_i^*}\tilde{b}_i^t &= \sum_{r=1}^{t-1}\tilde{\partial}_{\tilde{b}_i^*} \tilde{u}_i^{r}\sum_{j=1}^d X_{ij}^2 \tilde{\partial}_{\tilde{h}_j^r} \tilde{v}_{j}^{t-1},
\end{align*}
with $\tilde{\partial}_{\tilde{b}_i^*}\tilde{b}_i^* = 1, \tilde{\partial}_{\tilde{b}_i^s}\tilde{b}_i^* = 0$, and $\tilde{\partial}_{\tilde{b}_i^s} \tilde{u}_i^r,\tilde{\partial}_{\tilde{b}_i^*} \tilde{u}_i^r$ defined inductively by applying chain rule as in \Cref{rem:meaning-of-field-derivatives-in-dmft}. The operator $\tilde{\partial}_{\tilde{h}_j^s}$ is also defined analogously. Our first simplification is that we can replace ${b}_i^t,{h}_j^t$  by $\tilde{b}_i^t,\tilde{h}_j^t$, respectively.
\begin{claim}\label{claim:recursive-claim-1}
    For every $t,p\geq 1,i\in [N]$ and $j\in [d],$ we have 
    \begin{align*}
        \lim \norm{b_i^t - \tilde{b}_i^t} &= 0 \\
        \lim \norm{h_j^t - \tilde{h}_j^t} &= 0.
    \end{align*}
\end{claim}
\begin{proof}[of \Cref{claim:recursive-claim-1}]
      The proof is again by induction. At each step of the induction, we assume that for all $s,t\leq T-1$, all $p\geq 1$ and all $i\in [N],j\in [d],$ we have proved that
    \begin{align*}
        \max\left\{\norm{ h_j^t - \tilde{ h}_j^t}_p , \norm{\partial_{ h_j^s} h_j^t - \tilde{\partial}_{\tilde{ h}_j^s}\tilde{ h}_j^t}_p\right\}&\lesssim 1/\sqrt{N} \\
        \max\left\{\norm{ v_j^t - \tilde{ v}_j^t}_p , \norm{\partial_{ h_j^s} v_j^t - \tilde{\partial}_{\tilde{ h}_j^s}\tilde{ v}_j^t}_p\right\}&\lesssim 1/\sqrt{N} \\
        \max\left\{\norm{ b_i^t - \tilde{ b}_i^t}_p , \norm{\partial_{ b_i^s} b_i^t - \tilde{\partial}_{\tilde{ b}_i^s}\tilde{ b}_i^t}_p,\norm{\partial_{ b^*_i} b_i^t - \tilde{\partial}_{\tilde{ b}^*_i}\tilde{ b}_i^t}_p \right\}&\lesssim 1/\sqrt{N} \\
        \max\left\{\norm{ u_i^t - \tilde{ u}_i^t}_p , \norm{\partial_{ b_i^s} u_i^t - \tilde{\partial}_{\tilde{ b}_i^s}\tilde{ u}_i^t}_p,\norm{\partial_{ b_i^*} u_i^t - \tilde{\partial}_{\tilde{ b}^*_i}\tilde{ u}_i^t}_p \right\}&\lesssim 1/\sqrt{N}
    \end{align*}
    and aim to prove the same for $s<t=T$ (since in the case $s=t$, all partial derivatves are equal to 1 so the statement is trivial). Note that the base case $T=0$ is trivially true. Since all inductive cases are analogous, we elaborate only the cases of $\bfh$ and $\bfv.$ For $\bfh$, we have
    \begin{align*}
        \norm{h_j^t - \tilde{h}_j^t}_p &= \norm{\sum_{i=1}^N X_{ij}^2 \left(\tilde{\theta}_j \tilde{\partial}_{\tilde{b}_i^*} \tilde{u}_i^t - \theta_j \partial_{b_i^*}u_{i/ij}^t + \sum_{s=1}^t (\tilde{v}_j^{s-1} \tilde{\partial}_{\tilde{b}_i^s} \tilde{u}_i^t - v_{j/ij}^{s-1} \partial_{b_i^s}u_{i/ij}^t )\right)}_p + \norm{\Delta_j^t}_p \\
        &\lesssim  \frac{1}{N}\sum_{i=1}^N\left(\norm{ \tilde{\theta}_j \tilde{\partial}_{\tilde{b}_i^*} \tilde{u}_i^t - \theta_j \partial_{b_i^*}u_{i/ij}^t }_{2p}+ \sum_{s=1}^t \norm{\tilde{v}_j^{s-1} \tilde{\partial}_{\tilde{b}_i^s} \tilde{u}_i^t - v_{j/ij}^{s-1} \partial_{b_i^s}u_{i/ij}^t )}_{2p}\right) + \norm{\Delta_j^t}_p.
    \end{align*}
    We can bound each term individually. For $\norm{\Delta_j^t}_p$, using \Cref{thm:derivative-bounds}, we have 
    \begin{align*}
        \norm{\Delta_j^t}_p &\lesssim  \sum_{s=1}^{t-1} \sum_{i=1}^N \norm{u^{s}_{i/ij}X_{ij}^2\partial_{h_j^s}u_{i/ij}^{t}   }_p+ \frac{1}{2} \int_0^1 d\eta(1-\eta) \sum_{i=1}^N \norm{X_{ij}^3 \partial^2_{X_{ij}}u_i^{t}(X^{ij}(\eta))}_p \\
        &\lesssim \sum_{s=1}^{t-1} \sum_{i=1}^N \frac{1}{N}\norm{\partial_{h_j^s}u_{i/ij}^{t}   }_{3p}+ \frac{1}{2} \int_0^1 d\eta(1-\eta) \sum_{i=1}^N \frac{1}{N^{3/2}}\norm{\partial^2_{X_{ij}}u_i^{t}(X^{ij}(\eta))}_{2p} \\
         &\lesssim \sum_{s=1}^{t-1} \sum_{i=1}^N \frac{1}{N^{3/2}}+ \frac{1}{2} \int_0^1 d\eta(1-\eta) \sum_{i=1}^N \frac{1}{N^{3/2}} \\
         &\lesssim 1/\sqrt{N}.
    \end{align*}
    Similarly, for $\norm{ \tilde{\theta}_j \tilde{\partial}_{\tilde{b}_i^*} \tilde{u}_i^t - \theta_j \partial_{b_i^*}u_{i/ij}^t }_{2p}$ we have
    \begin{align*}
        \norm{ \tilde{\theta}_j \tilde{\partial}_{\tilde{b}_i^*} \tilde{u}_i^t - \theta_j \partial_{b_i^*}u_{i/ij}^t }_{2p} &= \norm{ \theta_j \tilde{\partial}_{\tilde{b}_i^*} \tilde{u}_i^t - \theta_j \partial_{b_i^*}u_{i/ij}^t }_{2p} \\
        &\lesssim \norm{  \tilde{\partial}_{\tilde{b}_i^*} \tilde{u}_i^t - \partial_{b_i^*}u_{i/ij}^t }_{4p} \\
        &\leq \norm{  \tilde{\partial}_{\tilde{b}_i^*} \tilde{u}_i^t - \partial_{b_i^*}u_{i}^t }_{4p} + \norm{\partial_{b_i^*}u_{i}^t - \partial_{b_i^*}u_{i/ij}^t}_{4p}.
    \end{align*}
    The first term is $O( 1/\sqrt{N})$ by inductive hypothesis. For the second term, we can write 
    \begin{align}
        \norm{\partial_{b_i^*}u_{i}^t - \partial_{b_i^*}u_{i/ij}^t}_{4p} &\leq \int_0^1d\eta \norm{X_{ij} \partial_{X_{ij}}\partial_{b_i^*}u_{i}^t(X^{ij}(\eta)}_{4p}\label{eq:taylor-expansion-to-put-xij-back} \\
        &\lesssim 1/\sqrt{N}\nonumber
    \end{align}
    again by \Cref{thm:derivative-bounds}. A similar proof to the above shows that 
    \begin{align*}
        \norm{\tilde{v}_j^{s-1} \tilde{\partial}_{\tilde{b}_i^s} \tilde{u}_i^t - v_{j/ij}^{s-1} \partial_{b_i^s}u_{i/ij}^t )}_{2p} &\lesssim 1/\sqrt{N}
    \end{align*}
    as well.

    Next we show that we have 
    \begin{align*}
        \norm{\partial_{ h_j^s} h_j^t - \tilde{\partial}_{\tilde{ h}_j^s}\tilde{ h}_j^t}_p &\lesssim 1/\sqrt{N}.
    \end{align*}
    Here we can write 
    \begin{align}
        &\norm{\partial_{ h_j^s} h_j^t - \tilde{\partial}_{\tilde{ h}_j^s}\tilde{ h}_j^t}_p \nonumber\\
        &\leq \norm{\partial_{ h_j^s} g_j^t}_p + \norm{\sum_{i=1}^N X_{ij}^2\left(\tilde{\partial}_{\tilde{ h}_j^s}\tilde{\theta}_j \tilde{\partial}_{\tilde{b}_i^*} \tilde{u}_i^t - \partial_{ h_j^s}\theta_j \partial_{b_i^*}u_{i/ij}^t + \sum_{s=1}^t (\tilde{\partial}_{\tilde{ h}_j^s}\tilde{v}_j^{s-1} \tilde{\partial}_{\tilde{b}_i^s} \tilde{u}_i^t - \partial_{ h_j^s}v_{j/ij}^{s-1} \partial_{b_i^s}u_{i/ij}^t )\right)}_p \nonumber\\
        &\qquad \qquad + \norm{\sum_{i=1}^N X_{ij}^2\left( \theta_j \partial_{ h_j^s}\partial_{b_i^*}u_{i/ij}^t + \sum_{s=1}^t v_{j/ij}^{s-1} \partial_{ h_j^s}\partial_{b_i^s}u_{i/ij}^t \right)}_p + \norm{\partial_{ h_j^s} \Delta_j^t}_p \nonumber\\
        &\lesssim \norm{\partial_{ h_j^s} g_j^t}_p + \frac{1}{N}\sum_{i=1}^N \left(\norm{\tilde{\partial}_{\tilde{ h}_j^s}\tilde{\theta}_j \tilde{\partial}_{\tilde{b}_i^*} \tilde{u}_i^t - \partial_{ h_j^s}\theta_j \partial_{b_i^*}u_{i/ij}^t}_{2p} + \sum_{s=1}^t \norm{\tilde{\partial}_{\tilde{ h}_j^s}\tilde{v}_j^{s-1} \tilde{\partial}_{\tilde{b}_i^s} \tilde{u}_i^t - \partial_{ h_j^s}v_{j/ij}^{s-1} \partial_{b_i^s}u_{i/ij}^t }_{2p} \right) \nonumber\\
        &\qquad \qquad + \frac{1}{N}\sum_{i=1}^N\left(\norm{ \theta_j \partial_{ h_j^s}\partial_{b_i^*}u_{i/ij}^t}_{2p} + \sum_{s=1}^t\norm{ v_{j/ij}^{s-1} \partial_{ h_j^s}\partial_{b_i^s}u_{i/ij}^t }_{2p} \right)+ \norm{\partial_{ h_j^s} \Delta_j^t}_p \label{eq:h-partials-difference} 
    \end{align}
    The second and last terms above can be bounded in a completely analogous way to how we bounded the corresponding terms in $\norm{h_j^t - \tilde{h}_j^t}_p$. The summands in the third term above can similarly be bounded since e.g.
    \begin{align*}
        \norm{ \theta_j \partial_{ h_j^s}\partial_{b_i^*}u_{i/ij}^t}_{2p} &\lesssim \norm{  \partial_{ h_j^s}\partial_{b_i^*}u_{i/ij}^t}_{4p}\\
        &\lesssim 1/\sqrt{N}
    \end{align*}
    by \Cref{thm:derivative-bounds}. What remains is to bound the first term $ \norm{\partial_{ h_j^s} g_j^t}_p.$ To bound this term, we use the same Taylor expansion as in \eqref{eq:V-FTC}. In this case we set 
    \begin{align}\label{eq:V-expansion-new}
    V_{i_1,\dots,i_p}(X) &:= (\partial_{ h_j^s} u_{i_1/i_1j}^{t-1}) \cdots (\partial_{ h_j^s} u_{i_p/i_pj}^{t-1}).
\end{align}
Following the same argument as in the proof of \Cref{thm:derivative-bounds}, we obtain 
\begin{align}
    &\norm{\partial_{h_j^s} g_j^t}_p^p \nonumber\\
    &\lesssim  \sum_{q=0}^p \int_{[0,1]^q} \prod_{\ell=1}^q d\eta_\ell\frac{1}{N^{(p+q)/2}}\sum_{(i_1,\dots,i_p) \in \calA_q} \norm{\partial_{X_{i'_1j}}\dots \partial_{X_{i'_qj}} V_{i_1,\dots,i_p}(X^{\{i'_1,\dots,i'_q\},j}(\eta_1,\dots,\eta_q))}_{p+1}. \nonumber
\end{align}
We claim that, for each of the terms in the sum above, we have 
\[
\norm{\partial_{X_{i'_1j}}\dots \partial_{X_{i'_qj}} V_{i_1,\dots,i_p}(X^{\{i'_1,\dots,i'_q\},j}(\eta_1,\dots,\eta_q))}_{p+1} \lesssim (1/\sqrt{N})^p.
\]
To prove this, it suffices to show that, for any way to distribute the partial derivatives $\partial_{X_{i'_1j}},\dots, \partial_{X_{i'_qj}}$ among the factors in $(\partial_{ h_j^s} u_{i_1/i_1j}^{t-1}), \cdots ,(\partial_{ h_j^s} u_{i_p/i_pj}^{t-1})$, every factor must have moments that are of order  $1/\sqrt{N}$. Note first that, since $X_{i'_k,j}$ has been set to zero in $\partial_{h_j^s}u_{i_k'/i_k'j}^{t-1},$ we can assume that no derivative $\partial_{X_{i'_k,j}}$ gets distributed to $\partial_{h_j^s}u_{i_k'/i_k'j}^{t-1},$ since otherwise the corresponding term will be equal to zero. But recall that the $i'_1,\dots,i'_k$ are unique among $i_1,\dots,i_p.$ Hence, any factors that get any derivatives at all distributed to them must have moments that are $\lesssim 1/\sqrt{N}$ by \Cref{thm:derivative-bounds}. Morever, any term that doesn't get any derivatives distributed to it is still being acted on by $\partial_{h_j^s}$, and hence again has moments that are $O(1/\sqrt{N})$ by \Cref{thm:derivative-bounds}. This proves that
\begin{align*}
    \norm{\partial_{h_j^s} g_j^t}_p^p &\lesssim \sum_{q=0}^p \frac{1}{N^{(p+q)/2}} \cdot |\calA_q|\cdot \frac{1}{N^{p/2}} \\
    &\lesssim \frac{1}{N^{p/2}}.
\end{align*}
Plugging this back into \eqref{eq:h-partials-difference}, we obtain 
\begin{align*}
    \norm{ \partial_{h_j^s}h_j^t - \tilde{\partial}_{\tilde{h}_j^s}\tilde{h}_j^t }_p &\lesssim 1/\sqrt{N},
\end{align*}
as desired.

To conclude the proof of \Cref{claim:recursive-claim-1}, we elaborate the inductive case for $\bfv.$ By Lipschitzness of $G_t$, we have 
\begin{align*}
    \norm{v_j^t - \tilde{v}_j^t}_p &\lesssim \sum_{s=1}^{t}\norm{h_j^s - \tilde{h}_j^s}_p \\
    &\lesssim 1/\sqrt{N}
\end{align*}
by inductive hypothesis. Similarly, by chain rule we have we have 
\begin{align*}
&\norm{\partial_{h_j^s} v_j^t - \tilde{\partial}_{\tilde{h}_j^s} \tilde{v}_j^t}_p \\
&= \sum_{k=1}^t \norm{\partial_k G_t(h_j^1, \dots, h_j^t; v_j^0,\tilde{w}_j^1,\dots, \tilde{w}_j^t) \partial_{h_j^s} h_j^k - \partial_k G_t(\tilde{h}_j^1, \dots, \tilde{h}_j^t; v_j^0,\tilde{w}_j^1,\dots, \tilde{w}_j^t) \tilde{\partial}_{\tilde{h}_j^s} \tilde{h}_j^k  }_p \\
&\leq \sum_{k=1}^t \norm{\partial_k G_t(h_j^1, \dots, h_j^t; v_j^0,\tilde{w}_j^1,\dots, \tilde{w}_j^t) \partial_{h_j^s} h_j^k - \partial_k G_t(\tilde{h}_j^1, \dots, \tilde{h}_j^t; v_j^0,\tilde{w}_j^1,\dots, \tilde{w}_j^t) \partial_{h_j^s} h_j^k}_p \\
&+ \sum_{k=1}^t \norm{\partial_k G_t(\tilde{h}_j^1, \dots, \tilde{h}_j^t; v_j^0,\tilde{w}_j^1,\dots, \tilde{w}_j^t) \partial_{h_j^s} h_j^k- \partial_k G_t(\tilde{h}_j^1, \dots, \tilde{h}_j^t; v_j^0,\tilde{w}_j^1,\dots, \tilde{w}_j^t) \tilde{\partial}_{\tilde{h}_j^s} \tilde{h}_j^k  }_p  \\
&\lesssim \sum_{k=1}^t \norm{\partial_k G_t(h_j^1, \dots, h_j^t; v_j^0,\tilde{w}_j^1,\dots, \tilde{w}_j^t) - \partial_k G_t(\tilde{h}_j^1, \dots, \tilde{h}_j^t; v_j^0,\tilde{w}_j^1,\dots, \tilde{w}_j^t) }_{2p}\\
&+ \sum_{k=1}^t \norm{ \partial_{h_j^s} h_j^k-  \tilde{\partial}_{\tilde{h}_j^s} \tilde{h}_j^k  }_{2p} \\
&\lesssim \sum_{k=1}^t\sum_{k'=1}^t \norm{h_j^{k'} - \tilde{h}_j^{k'}}_{2p} + \sum_{k=1}^t \norm{ \partial_{h_j^s} h_j^k-  \tilde{\partial}_{\tilde{h}_j^s} \tilde{h}_j^k  }_{2p} \\ 
&\lesssim 1/\sqrt{N}
\end{align*} 
by \Cref{thm:derivative-bounds} and the inductive hypothesis. This concludes the proof of \Cref{claim:recursive-claim-1}.
\end{proof}

\paragraph{Concentration of the coefficients.} We now do the second simplification that will lead us to prove \Cref{lem:inductive-case-lemma1}. We will define yet another set of variables, this time denoted with a bar, which will be equal to the variables with a tilde above, except they will have the terms $\sum_{i=1}^N X_{ij}^2\partial_{b_i^*}u_i^{t}  (X^{ij})$ and $\sum_{i=1}^N X_{ij}^2\partial_{b_i^s}u_i^{t}  (X^{ij})$ replaced by their limiting expectation. 

We set $\overline{\bftheta}=\bftheta,\overline{\bfv}^0=\bfv^0,\overline{\bfb}^* = \bfb^*.$ Moreover, for $t=1,\dots,T,$ we let
\begin{align*}
    \begin{cases}
         \overline{b}_i^t = z_j^t + \sum_{s=1}^{t-1}  \overline{u}_i^s \alpha \E_{\nu_{T-1}}  {\partial}_{ {h}^s}  {v}^{t-1} &\forall i\in [N]\\
         \overline{u}_i^t = F_t( \overline{b}_i^*,  \overline{b}_i^1,\dots, \overline{b}_i^t;w^1_i,\dots,w^t_i) &\forall i \in [N] \\
         \overline{h}_j^t = g_j^t +  \overline{\theta}_j  \E_{\mu_T}  {\partial}_{ {b}^*}  {u}^t + \sum_{s=1}^t  \overline{v}_j^{s-1} \E_{\mu_T}{\partial}_{ {b}^s}  {u}^t & \forall j\in [d] \\
         \overline{v}_j^t = G_t( \overline{h}_j^1,\dots, \overline{h}_j^t; \overline{v}_j^0, \tilde{w}^1_j,\dots, \tilde{w}^t_j) & \forall j\in [d]
    \end{cases},
\end{align*}
where the partial derivatives $\overline{\partial}$ are defined in the same way as $\tilde{\partial}.$ As we show next, the variables with a bar and with a tilde are close.
\begin{claim}\label{claim:recursive-claim-2}
    Recall we are assuming all links in the chain of \Cref{sec:induction-outline} are valid up to and not including \eqref{eq:nu-marginal-covergence}, \eqref{eq:nu-decoupling} at step $T.$ The following hold for every $i\in [N]$ and $j\in [d],$  and $t\leq T$: 
    \begin{align*}
        \lim \norm{ \tilde{b}_i^t - \overline{b}_i^t } _2&= 0 \\
        \lim \norm{\tilde{h}_j^t - \overline{h}_j^t}_2 &= 0. 
    \end{align*}
\end{claim}
\begin{proof} The proof is again by induction in $t=1,\dots,T.$ At each step, we assume that, for all times $s< t,$ we have
    \begin{align*}
        \lim \norm{ \tilde{b}_i^s - \overline{b}_i^s } _2&= 0 \\
        \lim \norm{\tilde{h}_j^s - \overline{h}_j^s}_2 &= 0,
    \end{align*}
and we aim to prove this at time $t$, as long as $t\leq T.$ Since the proofs of both of the statements above are analogous, we carry out explicitly the proof of the first claim only.

As a first step, note that by construction of the variables with tilde above, for $s\geq 1$ and $i \in [N]$, the quantities $\tilde{\partial}_{\tilde{h}_j^s} \tilde{v}_j^{t-1}$
    can be written as (recursively-defined) $C^\infty_{1b}$ functions of the fields $\{\tilde{h}_j^r\}_{r=1}^{s}.$ Call these functions $\phi_t$, such that $\tilde{\partial}_{\tilde{h}_j^s} \tilde{v}_j^{t-1} = \phi_t(\tilde{h}_j^1,\dots,\tilde{h}_j^{t-1}).$ Note that for $j\neq j' \in [d]$, we have
    \begin{align*}
        \Cov(\tilde{\partial}_{\tilde{h}_j^s} \tilde{v}_j^{t-1}, \tilde{\partial}_{\tilde{h}_{j'}^s} \tilde{v}_{j'}^{t-1}) &= \Cov(\phi_t(\tilde{h}_j^1,\dots,\tilde{h}_j^{t-1}), \phi_t(\tilde{h}_{j'}^1,\dots,\tilde{h}_{j'}^{t-1})) \\
        &\leq \Cov(\phi_t(h_j^1,\dots,h_j^{t-1}), \phi_t(h_{j'}^1,\dots,h_{j'}^{t-1})) + O(1/\sqrt{N}) \\
        &\lesssim 1/\sqrt{N}
    \end{align*}
    where in the first inequality we used \Cref{claim:recursive-claim-1} and in the second one we used \eqref{eq:nu-decoupling} at step $t-1$.

    Now we have 
    \begin{align*}
        \Var\left(\sum_{j=1}^d X_{ij}^2\tilde{\partial}_{\tilde{h}_j^s} \tilde{v}_j^{t-1} \right) &\lesssim \frac{1}{N} + N^2\max_{j\neq j'}\Cov(X_{ij}^2\tilde{\partial}_{\tilde{h}_j^s} \tilde{v}_j^{t-1}, X_{ij'}^2\tilde{\partial}_{\tilde{h}_{j'}^s} \tilde{v}_{j'}^{t-1})
    \end{align*}
    and, for $j\neq j',$
    \begin{align*}
    N^2&\Cov(X_{ij}^2\tilde{\partial}_{\tilde{h}_j^s} \tilde{v}_j^{t-1}, X_{ij'}^2\tilde{\partial}_{\tilde{h}_{j'}^s} \tilde{v}_{j'}^{t-1})  \leq N^2|\Cov(X_{ij}^2\tilde{\partial}_{\tilde{h}_j^s} \tilde{v}_j^{t-1}, \frac{1}{N}\tilde{\partial}_{\tilde{h}_j^s} \tilde{v}_j^{t-1})| \\
    &\qquad \qquad\qquad\qquad\qquad+ |\Cov(\tilde{\partial}_{\tilde{h}_j^s} \tilde{v}_j^{t-1}, \tilde{\partial}_{\tilde{h}_{j'}^s} \tilde{v}_{j'}^{t-1})| +N^2 |\Cov(\frac{1}{N}\tilde{\partial}_{\tilde{h}_{j'}^s} \tilde{v}_{j'}^{t-1}, X_{ij'}^2\tilde{\partial}_{\tilde{h}_{j'}^s} \tilde{v}_{j'}^{t-1})|.
    \end{align*}
    The middle term vanishes as stated above, so we need to control the first and third term, which are analogous, so we focus on the first term. We have 
    \begin{align*}
        N^2|\Cov(X_{ij}^2\tilde{\partial}_{\tilde{h}_j^s} \tilde{v}_j^{t-1}, \frac{1}{N}\tilde{\partial}_{\tilde{h}_j^s} \tilde{v}_j^{t-1})| &\lesssim 1/\sqrt{N} +N^2 |\Cov(X_{ij}^2{\partial}_{{h}_j^s} {v}_j^{t-1}, \frac{1}{N}{\partial}_{{h}_j^s} {v}_j^{t-1})| \\
        &\lesssim 1/\sqrt{N} +N^2 |\Cov(X_{ij}^2{\partial}_{{h}_j^s} {v}_{j/ij}^{t-1}, \frac{1}{N}{\partial}_{{h}_j^s} {v}_{j/ij}^{t-1})| \\
        &=1/\sqrt{N} + |\Cov({\partial}_{{h}_j^s} {v}_{j/ij}^{t-1}, {\partial}_{{h}_j^s} {v}_{j/ij}^{t-1})| \\
        &\lesssim 1/\sqrt{N} + |\Cov({\partial}_{{h}_j^s} {v}_{j}^{t-1}, {\partial}_{{h}_j^s} {v}_{j}^{t-1})| \\
        &\lesssim 1/\sqrt{N} + |\Cov(\tilde{\partial}_{\tilde{h}_j^s} \tilde{v}_{j}^{t-1}, \tilde{\partial}_{\tilde{h}_j^s} \tilde{v}_{j}^{t-1})| \\
        &\lesssim 1/\sqrt{N} + |\Cov(\phi_t(h_j^1,\dots,h_j^{t-1}), \phi_t(h_{j'}^1,\dots,h_{j'}^{t-1}))| \\
        &\to 0,
    \end{align*}
    where we have again used \Cref{claim:recursive-claim-1} and the inductive hypothesis \eqref{eq:nu-decoupling}. This shows that
    \begin{align}\label{eq:var1-vanish}
        \Var\left(\sum_{j=1}^d X_{ij}^2\tilde{\partial}_{\tilde{h}_j^s} \tilde{v}_j^{t-1} \right)   \to 0.
    \end{align}
    Moreover, by inductive hypothesis, we have 
    \begin{align}
        \left|\E\sum_{j=1}^d X_{ij}^2\tilde{\partial}_{\tilde{h}^s_j} \tilde{v}_j^{t-1}  -\alpha\E_{\nu_{T-1}} {\partial}_{{h}^s} {v}^{t-1} \right| &\lesssim \left|\E\sum_{j=1}^d X_{ij}^2{\partial}_{{h}_j} {v}_j^{t-1}  - \frac{1}{N}\sum_{j=1}^d \E {\partial}_{{h}_j} {v}_j^{t-1} \right|  + 1/\sqrt{N}\nonumber\\
        &\lesssim  \left|\E\sum_{j=1}^d X_{ij}^2{\partial}_{{h}_j} {v}_{j/ij}^{t-1}  - \frac{1}{N}\sum_{j=1}^d \E {\partial}_{{h}_j} {v}_{j/ij}^{t-1} \right|  + 1/\sqrt{N}\nonumber\\
        &= 1/\sqrt{N}\label{eq:mean1-vanish}
    \end{align}
    From \eqref{eq:var1-vanish}, \eqref{eq:mean1-vanish} and the fact that 
    \begin{align*}
        \norm{\tilde{u}_i^s - \overline{u}_i^s}_2 &= \norm{F_t( \tilde{b}_i^*,  \tilde{b}_i^1,\dots, \tilde{b}_i^t;w^1_i,\dots,w^t_i) - F_t( \overline{b}_i^*,  \overline{b}_i^1,\dots, \overline{b}_i^t;w^1_i,\dots,w^t_i)}_2\\
        &\to 0
    \end{align*}
    by inductive hypothesis, we conclude $\norm{\tilde{b}_i^t - \overline{b}_i^t}_2 \to 0$, as desired.
\end{proof}

\paragraph{Conclusion of Proof of \Cref{lem:inductive-case-lemma1}.}

To conclude the proof of \Cref{lem:inductive-case-lemma1}, note that combining \Cref{claim:recursive-claim-1} and \Cref{claim:recursive-claim-2}, we have, for all $i \in [N],j\in [d],$
\begin{align*}
    \lim \norm{ b_i^t - \overline{b}_i^t } _2&= 0 \\
    \lim \norm{h_j^t - \overline{h}_j^t}_2 &= 0. 
\end{align*}
Moreover, from the definition of $\overline{h}_j^t,$ we have that there exists a (recursively defined) function $\tilde{\phi}$ such that
\begin{align*}
    \phi(\{\overline{h}_j^t\}_{t=1}^T;v_j^0,\theta_j,\{\tilde{w}_j^t\}_{t=1}^T) &= \tilde{\phi}(\{g_j^t\}_{t=1}^T;v_j^0,\theta_j,\{\tilde{w}_j^t\}_{t=1}^T)
\end{align*}
This proves \eqref{eq:gjt-mu-t-close}. Moreover, note that by construction we also have $\E_{\overline{\nu}_T} \phi( \{h^t\}_{t=1}^T;v^0,\theta,\{\tilde{w}^t\}_{t=1}^T),$ as desired.

\subsubsection{Proof of \Cref{lem:inductive-case-lemma2}}
Let $\{\tilde{X}_{ij}^t\}_{i\in [N], j\in [d], t\in [T]}$ be centered, jointly Gaussian random variables with covariance $\E \tilde{X}_{ij}^s\tilde{X}_{i'j'}^t = \delta_{ii'}\delta_{jj'}\frac{1}{N}\E u_i^{s}u_i^{t
},$ independent of $X$ and of everything else, and define 
    \[
    \tilde{g}_j^t = \sum_{i=1}^N \tilde{X}_{ij}^t.
    \]
We will prove proving that, for every $C^\infty_{1b}$ function $\Psi:\R^{2\times (2T+2)}\to\R$, and every $j\neq j',$ we have
    \begin{align}
            \lim &|\E\Psi( \{g_j^t\}_{t=1}^T,\{g_{j'}^t\}_{t=1}^T ;v^0_j,\theta_j,\{\tilde{w}_j^t\}_{t=1}^T, v^0_{j'},\theta_{j'},\{\tilde{w}_{j'}^t\}_{t=1}^T)\nonumber \\
            &\qquad \qquad \qquad - \E\Psi( \{\tilde{g}_j^t\}_{t=1}^T,\{\tilde{g}_{j'}^t\}_{t=1}^T ;v^0_j,\theta_j,\{\tilde{w}_j^t\}_{t=1}^T, v^0_{j'},\theta_{j'},\{\tilde{w}_{j'}^t\}_{t=1}^T) | \to 0. \nonumber
    \end{align}
    From now on, to avoid clutter, we wrote the above as 
    \begin{align}
            \lim &|\E\Psi( \{g_j^t\}_{t=1}^T,\{g_{j'}^t\}_{t=1}^T )- \E\Psi( \{\tilde{g}_j^t\}_{t=1}^T,\{\tilde{g}_{j'}^t\}_{t=1}^T) | \to 0,\label{eq:can-replace-with-gaussian}
    \end{align}
    making the dependence on the other variables implicit in the notation. Note that, once we have \eqref{eq:can-replace-with-gaussian}, choosing $\Psi(\{x^t\}_{t=1}^T, \{y^t\}_{t=1}^T) = \phi(\{x^t\}_{t=1}^T)$ and $\Psi(\{x^t\}_{t=1}^T, \{y^t\}_{t=1}^T) = \phi(\{x^t\}_{t=1}^T)\phi(\{y^t\}_{t=1}^T)$ shows that it suffices for us to show \eqref{eq:gt-marginal-convergence}, \eqref{eq:gt-decoupling} with $\tilde{g}$ in place of $g.$ But \eqref{eq:gt-decoupling} is trivially true for $\tilde{g}$ since $\tilde{g}^t_j$ and $\tilde{g}^t_{j'}$ are independent. Moreover, by exchangeability we have 
    \begin{align*}
        \E\tilde{g}_j^s \tilde{g}_j^t &= \frac{1}{N}\sum_{i=1}^N \E u_i^{s}u_i^{t} \\
        &= \E_{\mu_{T}} u^{s}u^{t}+ o(1)
    \end{align*}
    by inductive hypothesis, since $u_i^t$ is a $C^\infty_{1b}$ function of the fields $b^*_i, b^1_i,\dots,b^t_i$. This proves \eqref{eq:gt-marginal-convergence}.
    
    Hence, to conclude the proof it suffices to show \eqref{eq:can-replace-with-gaussian}. Similar to the base case, we consider an interpolation: for $0\leq k\leq N,$ we define
    \begin{align*}
        S_j^{t,k} &= \sum_{i=1}^k X_{ij} u_{i/ij}^{t} + \sum_{i=k+1}^N \tilde{X}_{ij}^t \\
        S_j^{t,-k} &= \sum_{i=1}^{k-1} X_{ij} u_{i/ij}^{t}+ \sum_{i=k+1}^N \tilde{X}_{ij}^t 
    \end{align*}
    so that $S_j^{t,N}=g_j^t$ and $S_j^{t,0}=\tilde{g}_j^t$. We now need the following additional lemma, which is also proved at the end of the section. We state the lemma for the second moments because that's all we need, but the lemma in fact holds true for all moments.
    \begin{claim}\label{claim:recursive-claim-1}
        Let $S_{j}^{t,k}, S_{j}^{t,-k}$ be defined as above, with the $X_{ij}$ being any $\frac{\sigma}{\sqrt{N}}$-subgaussian random variables. Then, for every $t\geq 1$ independent of $N$ and $j,j' \in [d],k\in [N]$, including $j=j',$ we have 
        \begin{align*}
            \norm{S_j^{t,k}}_2 &\lesssim 1 \\
            \norm{\partial_{X_{kj'}}S_j^{t,-k}}_2 &\lesssim 1/\sqrt{N}.
        \end{align*}
    \end{claim}
    Telescoping, for $j\neq j'$, we have
    \begin{align*}
        |\E\Psi(\{g_j^t\}_{t=1}^T, \{g_{j'}^t\}_{t=1}^T)& - \E\Psi(\{\tilde{g}_j^t\}_{t=1}^T, \{\tilde{g}_{j'}^t\}_{t=1}^T)| \\
        &\leq \sum_{k=1}^N|\E\Psi(\{S_j^{t,k}\}_{t=1}^T, \{S_{j'}^{t,k}\}_{t=1}^T) - \E\Psi(\{S_j^{t,k-1}\}_{t=1}^T, \{S_{j'}^{t,k-1}\}_{t=1}^T)|.
    \end{align*}
    Now we expand the above around $\{S_j^{t,-k}\}_{t=1}^T, \{S_j^{t,-k}\}_{t=1}^T.$ To avoid clutter, below we write $\Psi^{-k} = \Psi(\{S_j^{t,-k}\}_{t=1}^T, \{S_j^{t,-k}\}_{t=1}^T)$ and, for $t=1,\dots,2T$,
    \begin{align*}
    j(t) &= \begin{cases}
        j \qquad &\text{if }t\leq T \\
        j' &\text{if }T<t\leq 2T.
    \end{cases}
    \end{align*}
    We have 
    \begin{align}
        |\E\Psi(\{S_j^{t,k}\}_{t=1}^T, \{S_{j'}^{t,k}\}_{t=1}^T) - &\E\Psi(\{S_j^{t,k-1}\}_{t=1}^T, \{S_{j'}^{t,k-1}\}_{t=1}^T)| \nonumber\\
        &\lesssim \sum_{t=1}^{2T} |\E (\partial_t \Psi^{-k})\cdot (X_{kj(t)}u_{k/kj(t)}^{t} - \tilde{X}^t_{kj(t)})| \nonumber\\
        &\qquad +\sum_{s,t=1}^{2T} | \E(\partial_{s}\partial_t \Psi^{-k})\cdot (X_{kj(s)}X_{kj(t)}u_{k/kj(s)}^{s}u_{k/kj(t)}^{t} - \tilde{X}^s_{kj(s)}\tilde{X}^t_{kj(t)})| \nonumber\\
        &\qquad + N^{-3/2},\label{eq:Psi-expansion}
    \end{align}
    where, to bound the last term by $\lesssim N^{-3/2}$, we have used \Cref{claim:recursive-claim-1} and the Lipschitzness of $\Psi.$ We now want to show each of the terms in the two sums above are also $\lesssim N^{-3/2}.$ We begin with a term in the first sum, assuming without loss of generality that $t\leq T.$ Since $\tilde{X}_{kj}$ is independent of everything else, we have 
    \begin{align*}
        |\E (\partial_t \Psi^{-k})\cdot (X_{kj}u_{k/kj}^{t} - \tilde{X}_{kj})| &=|\E (\partial_t \Psi^{-k})\cdot X_{kj} u_{k/kj}^{t}|.
    \end{align*}
    Expanding around $X_{kj}=0$ and using \Cref{claim:recursive-claim-1}, we have 
    \begin{align*}
        |\E (&\partial_t \Psi^{-k})\cdot X_{kj} u_{k/kj}^{t}| \\
        &\lesssim  |\E (\partial_t \Psi^{-k})(X^{\{kj\}})\cdot X_{kj} u_{k/kj}^{t}| + |\E ( \partial_{X_{kj}}\partial_t\Psi^{-k})(X^{\{kj\}})\cdot X^2_{kj} u_{k/kj}^{t}| + N^{-3/2}\\
        &\lesssim \frac{1}{N}\E| ( \partial_{X_{kj}}\partial_t\Psi^{-k})(X^{\{kj\}})| + N^{-3/2},
    \end{align*}
    where in the second step we have used that the first term is equal to zero. Now $\partial_t \Psi$ is a smooth function with all non-zero derivatives uniformly bounded. Hence
    \begin{align*}
        \E| ( \partial_{X_{kj}}\partial_t\Psi^{-k})(X^{\{kj\}})| &\lesssim \sum_{t=1}^T \norm{\partial_{X_{kj}}S^{t,-k}_j}_2 + \norm{\partial_{X_{kj}}S^{t,-k}_{j'}}_2 \\
        &\lesssim 1/\sqrt{N}
    \end{align*}
    by \Cref{claim:recursive-claim-1}. This proves that the terms in the first sum of \eqref{eq:Psi-expansion} are $\lesssim N^{-3/2}.$ For the second term, the situation is simpler: using \Cref{claim:recursive-claim-1}, we have
    \begin{align*}
        | \E(\partial_{s}\partial_t &\Psi^{-k})\cdot (X_{kj(s)}X_{kj(t)}u_{k/kj(s)}^{s}u_{k/kj(t)}^{t} - \tilde{X}^s_{kj(s)}\tilde{X}^t_{kj(t)})| \\
        &\lesssim | \E(\partial_{s}\partial_t \Psi^{-k}(X^{\{kj(s), kj(t)\}})\cdot (X_{kj(s)}X_{kj(t)}u_{k/kj(s)}^{s}u_{k/kj(t)}^{t} - \tilde{X}^s_{kj(s)}\tilde{X}^t_{kj(t)})| + N^{-3/2}\\
        &\lesssim \frac{1}{N}|\delta_{j(s)j(t)}\E[u_{k/kj(s)}^{s}u_{k/kj(t)}^{t}] - \delta_{j(s)j(t)} \E [u_k^{s}u_k^{t}]|  + N^{-3/2} \\
        &\lesssim N^{3/2},
    \end{align*}
    were we have used \Cref{thm:derivative-bounds} in the last step. This proves \eqref{eq:can-replace-with-gaussian}, completing the proof of \Cref{lem:inductive-case-lemma2} and hence of \Cref{thm:main-strong}.

To conclude, we prove \Cref{claim:recursive-claim-1}.
\begin{proof}[of \Cref{claim:recursive-claim-1}]
We have 
\begin{align*}
    \norm{S_j^{t,k}}_2 &\leq \norm{\sum_{i=1}^{k}X_{ij} u_{i/ij}^{t} }_2 + \norm{\sum_{i=k+1}^N \tilde{X}_{ij}^t}_2.
\end{align*}
The second term satistfies 
\begin{align*}
    \norm{\sum_{i=k+1}^N \tilde{X}_{ij}^t}_2^2 &= \frac{1}{N}\sum_{i=k+1}^N \norm{u_i^{t}}_2^2  \lesssim 1
\end{align*}
by \Cref{thm:derivative-bounds}. The first term on the other hand is bounded in the same way that we bounded 
\[
\norm{\sum_{i=1}^N X_{ij} \partial_{h_j^s} u_{i/ij}^{t}}_p^p
\]
in the proof of \Cref{claim:recursive-claim-1} above. Next, again using the same expansion as in the proof of \Cref{claim:recursive-claim-1}, we have
    \begin{align*}
        &\norm{\partial_{X_{kj'}}S^{t,-k}_j}_2^2 = \E\left(\sum_{i=1}^{k-1} X_{ij} \partial_{X_{kj'}}u_{i/ij}^{t}\right)^2 \\
        &= \sum_{i=1}^{k-1}\E X_{ij}^2( \partial_{X_{kj'}}u_{i/ij}^{t})^2 +  \sum_{1\leq i\neq i'<k}\E X_{ij}X_{i'j}( \partial_{X_{kj'}}u_{i/ij}^{t})( \partial_{X_{kj'}}u_{i'/i'j}^{t}) \\
        &\lesssim 1/N + \int_0^1 \int_0^1d\eta d\eta'\sum_{1\leq i\neq i'<k}\E X_{ij}^2X_{i'j}^2(\partial_{X_{ij}}\partial_{X_{i'j}})( \partial_{X_{kj'}}u_{i/ij}^{t})( \partial_{X_{kj'}}u_{i'/i'j}^{t})(X^{\{i,i'\},j}(\eta,\eta')) \\
         &= 1/N + \int_0^1 \int_0^1d\eta d\eta'\sum_{1\leq i\neq i'<k}\E X_{ij}^2X_{i'j}^2( \partial_{X_{i'j}}\partial_{X_{kj'}}u_{i/ij}^{t})( \partial_{X_{ij}}\partial_{X_{kj'}}u_{i'/i'j}^{t})(X^{\{i,i'\},j}(\eta,\eta')) \\
         &\lesssim 1/N
\end{align*}
by \Cref{thm:derivative-bounds}, since $i,i'$ and $k$ are all distinct, $i$ and $i'$ are isolated in the graphs of $\partial_{X_{i'j}}\partial_{X_{kj'}}$ and $\partial_{X_{ij}}\partial_{X_{kj'}}$, respectively. This concludes the proof.
\end{proof}

\section{Subgaussian and subexponential random variables}\label{sec:subgaussian-subexponential}
Here we state a few standard facts about subgaussian and subexpoential random variables that we use throughout our proofs. All proofs can be found in standard texts, e.g. \cite{vershynin2018high}. In the lemmas below, for simplicity and without loss of generality, we take all the absolute constants to be the same, denoted by $C.$

\begin{definition}\label{defn:subgaussian-subexponential-norms}
Let $X$ be a random variable in $\R.$ We call the quantities
\begin{align*}
    \norm{X}_{\psi_1} &:= \inf\{t>0 : \E e^{X/t}\leq 2\} \\
    \norm{X}_{\psi_2} &:= \inf\{t>0 : \E e^{X^2/t^2}\leq 2\}
\end{align*}
in $\R_+ \cup \{\infty\}$ the subexpotential and subgaussian norms of $X$, respectively.
\end{definition}

\begin{lemma}\label{lem:prod-of-subgaussians}
For every pair of random variables $X,Y$ in $\R$, we have
\begin{align*}
    \norm{XY}_{\psi_1} &\leq \norm{X}_{\psi_1}\norm{Y}_{\psi_2}.
\end{align*}
\end{lemma}

\begin{lemma}\label{lem:subgaussian-subexp-moments}
 There exists a universal constant $C>0$ such that the following holds. For every random variable $X$ in $\R$ and $p\geq 1,$ we have 
 \begin{align*}
     \norm{X}_{p} &\leq C \norm{X}_{\psi_1} p \\
     \norm{X}_{p} &\leq C \norm{X}_{\psi_2} \sqrt{p}.
 \end{align*}
\end{lemma}

\begin{lemma}\label{lem:sum-subgaussian-subexp}
    There exists a universal constant $C>0$ such that the following holds. Suppose $X_1,\dots,X_n$ are independent, centered random variables. Then we have 
    \begin{align*}
        \norm{\sum_{i=1}^n X_i}_{\psi_2}^2 &\leq C\sum_{i=1}^n \norm{X_i}_{\psi_2}^2 \\
        \norm{\sum_{i=1}^n X_i}_{\psi_1}^2 &\leq C\sum_{i=1}^n \norm{X_i}_{\psi_1}^2.
    \end{align*}
\end{lemma}